\documentclass[a4paper,reqno,x11names]{amsart}

\usepackage{mathtools}
\usepackage{mathabx, amscd}
\usepackage[utf8]{inputenc}
\usepackage{amsmath}
\usepackage{amsfonts}
\usepackage{amsthm}
\usepackage{amssymb}

\usepackage[ruled,vlined]{algorithm2e}
\usepackage{caption}
\usepackage[normalem]{ulem}
\usepackage{enumerate}
\usepackage{tikz}
\usepackage{pgfplots}
\usepackage{nicematrix}
\usetikzlibrary{matrix,arrows,decorations.pathmorphing}
\usetikzlibrary{patterns}
\usepackage{subcaption}
\usetikzlibrary{decorations.pathreplacing}
\usetikzlibrary{arrows.meta, positioning}
\usepackage{xcolor}
\usepackage[hidelinks,pdfencoding=auto]{hyperref}
\usepackage{mathrsfs}
\hypersetup{colorlinks,breaklinks,
            citecolor=cyan,
             urlcolor=blue,
             linkcolor=blue}
\usepackage{colortbl}
\usepackage{fullpage}
\usepackage{booktabs}
\usepackage{float}
\usepackage{graphicx}

\theoremstyle{plain}
\newtheorem{thm}{Theorem}[section]
\newtheorem{lem}[thm]{Lemma}
\newtheorem{pro}[thm]{Proposition}
\newtheorem{cor}[thm]{Corollary}
\newtheorem*{pb}{Problem}

\theoremstyle{definition}
\newtheorem{defn}[thm]{Definition}

\newtheorem{rem}[thm]{Remark}

\newtheorem{exa}[thm]{Example}
\newtheorem{nota}[thm]{Notation}

\usepackage[most]{tcolorbox}
\xdefinecolor{myblue}{rgb}{0.1,0.1,0.8}
\xdefinecolor{mycol1}{rgb}{0.01,0.7,0.7}
\xdefinecolor{mycol}{rgb}{0.5,0.01,0.01}
\xdefinecolor{mycol2}{rgb}{0.01,0.7,0.01}

\xdefinecolor{yell}{rgb}{0.03,0.04,0.04}
\xdefinecolor{orn}{rgb}{0.04,0.03,0.04}
\xdefinecolor{rd}{rgb}{0.04,0.04,0.03}

\newcommand{\rd}{\cellcolor{DarkSeaGreen2}}

\newcommand{\orn}{\cellcolor{Thistle2}}
\newcommand{\yell}{\cellcolor{LightSkyBlue1}}
\newcommand\gr{\cellcolor{LavenderBlush1}}
\newcommand\y{\cellcolor{Azure2}}

\renewcommand{\th}{\theta}
\renewcommand{\TH}{\boldsymbol{\th}}

\newcommand{\LL}{\mathbb{L}}

\newcommand{\F}{\mathbb{F}}
\newcommand{\Fq}{\F_q}
\newcommand{\Fqm}{\F_{q^m}}

\newcommand{\f}{\mathbf{f}}
\newcommand{\Q}{\mathbb{Q}}
\newcommand{\K}{\mathbb{K}}
\newcommand{\ZZ}{\mathbb{Z}}
\newcommand{\NN}{\mathbb{N}}

\renewcommand{\a}{\alpha}
\renewcommand{\b}{\beta}
\newcommand{\efar}{e_{\textrm{far}}}
\newcommand{\Dfar}{\Delta_{\textrm{far}}}

\newcommand{\code}[1]{\mathscr{#1}}
\newcommand{\CC}{\code{C}}

\newcommand{\RM}{\text{RM}}

\renewcommand{\aa}{\mathbf{a}}
\newcommand{\av}{\mathbf{a}}
\newcommand{\bb}{\mathbf{b}}
\newcommand{\cv}{\mathbf{c}}
\newcommand{\ev}{\mathbf{e}}

\newcommand{\nn}{\mathbf{n}}

\newcommand{\vv}{\mathbf{v}}

\newcommand{\yv}{\mathbf{y}}

\newcommand{\wt}{\text{w}_{\text{H}}}

\newcommand{\B}{\mathcal{B}}

\newcommand{\rank}{\text{Rk}\,}
\newcommand{\rk}{\text{Rk}\,}

\newcommand{\Am}{{\mathbf{A}}}

\newcommand{\Dm}{{\mathbf{D}}}

\newcommand{\Dick}{\mathbf{D}}

\newcommand{\Span}[2]{{\left\langle #1 \right\rangle}_{#2}}
\newcommand{\<}{\left<}
\renewcommand{\>}{\right>}

\newcommand{\Gal}{\text{Gal}}
\newcommand{\GG}{\textrm{G}}
\renewcommand{\gg}{\texttt{g}}
\newcommand{\hg}{\texttt{h}}
\newcommand{\hh}{\texttt{h}}

\newcommand{\ii}{\mathbf{i}}
\newcommand{\jj}{\mathbf{j}}

\newcommand{\eqdef}{\stackrel{\text{def}}{=}}
\newcommand{\map}[4]{
  \left\{
  \begin{array}{ccc}
    #1 & \longrightarrow & #2 \\
    #3 & \longmapsto     & #4
  \end{array}
  \right.
}
\newcommand{\pinv}{\varphi^{-1}}

\newcommand{\Id}{\mathrm{Id}}
\newcommand{\OO}{\mathcal{O}}

\colorlet{known}{teal}
\colorlet{unknown}{black}

  \def\longversion{0}

\title{Decoding rank metric Reed--Muller codes} \author{Alain Couvreur}
\author{Rakhi Pratihar} \address{Inria \& Laboratoire LIX, CNRS UMR
  7161, École Polytechnique, Institut Polytechnique de Paris, 1 rue
  Honoré d'Estienne d'Orves, 91120 Palaiseau Cedex}
\email{\{alain.couvreur,rakhi.pratihar\}@inria.fr } \thanks{This
  work was partially funded by the French Agence Nationale de la
  Recherche through the France 2030 ANR project ANR-22-PETQ-0008
  PQ-TLS. The authors are also partially funded by French Grant ANR
  projet \emph{projet de recherche collaboratif}
  ANR-21-CE39-0009-BARRACUDA, and by Horizon-Europe MSCA-DN project
  ENCODE.}
\begin{document}

\begin{abstract}
In this article, we investigate the decoding of the rank metric Reed--Muller codes introduced by Augot, Couvreur, Lavauzelle and Neri in 2021. These codes are defined from Abelian Galois extensions extending the construction of Gabidulin codes over arbitrary cyclic Galois extensions.
We propose a polynomial time algorithm that rests on the structure of Dickson matrices, works on any such code and corrects any error of rank up to half the minimum distance.
\end{abstract}
 \maketitle

\section{Introduction}
\emph{Rank metric codes} were introduced by Delsarte in 1978 \cite{Del} as a set of $m \times n$ matrices over a finite field $\Fq$ where $q$ is a prime power, with a combinatorial interest, where the rank distance between two matrices is measured by the rank of their difference. Thereafter, Gabidulin in 1985 \cite{Gab} and Roth in 1991 \cite{Roth}, independently defined a variant of rank metric codes as a set of vectors of length $n$ over a finite extension $\Fqm$ of $\Fq$, where the rank distance of two vectors in $\Fqm^n$ is given by the dimension of the $\Fq$-space spanned by the coordinates of their difference. In recent years, these codes have attracted significant attention due to their applications to network coding \cite{SKK08}, distributed data storage \cite{TC}, space-time coding \cite{GBL03, Puchinger}, code-based cryptography \cite[Chapter 3.2]{BHLPRW}. Rank metric code constructions can be extended to codes over infinite fields. For instance, mainly for crisscross error correction purpose, Roth presented such a construction over algebraically closed fields in \cite{Roth}, and as more general class as tensor codes with tensor-rank metric over arbitrary field extensions in \cite{Roth96}. In another line of works, the extension of the theory of rank metric codes from finite fields to arbitrary cyclic Galois extensions has been treated thoroughly in \cite{Aug, ALR13, ALR18} and thereafter, to arbitrary finite Galois extensions in \cite{ACLN}.

Finding families of rank metric codes with efficient
decoding algorithms is a problem of interest for various applications, for instance, error-correction in random network coding \cite{SKK08}, distributed data storage \cite{TC}, or cryptography where rank metric codes have been used, among others, to instantiate the GPT scheme \cite{GPT}. More specifically, rank metric codes over infinite fields are used in the field of image processing where the decoding problem is equivalent to the low-rank matrix recovery problem \cite{MPMB16} and also in space-time coding \cite{GBL03,Puchinger}. 
However, only few classes of rank metric codes with effective decoding algorithms are known; simple codes \cite{GHPT}, some families of \emph{maximum rank distance} (MRD) codes including Gabidulin codes and their variants, cf. \cite[Chapter 2]{BHLPRW}, and low-rank parity check (LRPC) codes \cite{RJB} and their interleaved version. Therefore, compared
  to the Hamming metric setting, rank metric still suffers from a lack
  of diversity in terms of families of codes equipped with efficient
  decoding algorithms. This question is of ever growing interest with
  the recent rise of new post--quantum cryptographic primitives where
  the need for efficient solvers of various decoding problems is
  recurrent. Indeed, on one hand, McEliece-like schemes \cite{M78} require codes
  with an efficient decoder and whose structure can be hidden to the
  attackers.  Such a scheme has been instantiated with codes in
  Hamming and rank metric.  On the other hand, many Alekhnovich-like
  schemes in code or lattice based cryptography require a decoder to
  conclude the decryption phase and get rid from a residual noise
  term. See for instance \cite{AABBBBDDGLPRVZ22,A19b}. For these
  reasons, there is a strong motivation in broadening the diversity of
  decodable codes in rank metric: first to improve our
  understanding of decoding problems and second in view toward
  applications to future new cryptographic designs.

Regarding codes over infinite fields, there are decoding algorithms for Gabidulin codes in characteristic zero \cite{ALR18,MPMB16,Rob16} and for optimal array codes over algebraically closed fields \cite{Roth17}. In this paper, we present an efficient new decoding algorithm for the \emph{rank metric Reed--Muller codes} introduced in \cite{ACLN} as subspaces of skew group algebra $\LL[\GG]$ for arbitrary Abelian Galois extension $\LL/\K$ with Galois group $\GG$. This provides a new class of rank metric codes over infinite fields with efficient decoding algorithm.

Rank metric Reed--Muller codes, also called $\TH$-Reed--Muller codes in \cite{ACLN}, where $\TH = (\theta_1, \ldots, \th_m)$ specifies a generating set of the Abelian group $\GG$ (see Definition \ref{theta_RM}), are a ``multivariate'' version of Gabidulin codes which are defined in the case where $\GG$ is cyclic.
The decoding techniques for Gabidulin codes and its variants are mainly syndrome-based decoding \cite{Gab, RP, Roth, GPT} and interpolation-based decoding \cite{KL, KLZ, Li, Loid, Tovo}. 
The decoding algorithm in \cite{Loid} in the finite fields setting is extended in
\cite{ALR18} to the general case. Loidreau \emph{et al.} approach
rests on a Welch--Berlekamp--like approach consisting of computing some
linear polynomial ``localizing'' the errors. This approach fits in the
general paradigm of \emph{error locating pairs} developed in the
Hamming setting by Pellikaan \cite{P92} and independently by K\"otter
\cite{K92}. This paradigm was extended to rank metric by Martinez--Pe\~{n}as and Pellikaan in \cite{MP}. The latter was used in \cite{ACLN} to decode rank metric Reed--Muller codes with rather limited decoding radii.

\medskip

\subsection*{Our Contribution.} In the present paper, we propose an alternative approach based on the use of $\GG$-Dickson matrices. The idea is to reconstruct a $\TH$-polynomial by recovering its coefficients in an iterative manner by a majority voting method on submatrices of the associated $\GG$-Dickson matrix. This majority voting procedure was originally formulated by Massey in \cite{Massey63} for decoding linear systematic codes and later adapted for various other classes of codes, \emph{e.g.}, see \cite{AGM76} and the survey on decoding algebraic geometric codes \cite{HP02}. It is worth mentioning that in the Hamming metric case the majority voting method is employed by transforming the Hamming error into a rank argument on a matrix of syndromes, while we will see that the $\GG$-Dickson matrix representation of the error $\TH$-polynomial already presents a rank constraint that enables to apply a majority voting to recover the unknown coefficients. To our knowledge, this article is the first use of majority voting for decoding rank metric codes. 

Rank metric Reed--Muller codes were introduced in \cite{ACLN} where a
first attempt of decoding was included using the rank analogue of
error correcting pairs \cite{MP}. Denoting by $\omega$ the complexity
exponent of classical linear algebra operations and by $N$ the degree
of the extension $\LL/\K$, the algorithm of \cite{ACLN} required a
complexity of $\OO(N^{2\omega})$ operations in $\K$ to achieve a
decoding radius that was far below half the minimum distance. In this
article, we propose a new algorithm that corrects any error pattern up
to half the minimum distance in $\widetilde{\OO}(N^4)$ operations in
$\K$.

 \if\longversion1
 We study the structure of the particular class of binary $\TH$-Reed-Muller codes, i.e., $\GG \cong (\ZZ/{2\ZZ})^m$. We observe that these codes possess a recursive structure, which can be considered as a rank analogue of the $(u\vert u+v)$--structure for binary Reed-Muller codes in Hamming metric. Depending on characteristic of the base field $\K$ equals to 2 or different from 2, the codes exhibit two different recursive structures. However, for both the cases, we exploit the recursive structure to give a probabilistic decoding algorithm for binary rank Reed--Muller codes that corrects error of rank up to half minimum distance. 
 \fi
 
\subsection*{Organization of the article}
Section \ref{sec:2} introduces the notations used in this paper, as well as basic notions regarding rank metric codes as subspaces of skew group algebras including rank metric Reed-Muller codes, and some properties of $\GG$-Dickson matrices. In Section \ref{sec:dickson_framework}, we give the framework for decoding rank metric Reed-Muller codes using $\GG$-Dickson matrices and we first illustrate this decoding approach for Gabidulin codes in Section \ref{sec:first_examples}. We then describe how this approach can be used for decoding Reed-Solomon codes. 
Section \ref{sec:Dickson_decoding} presents a decoding algorithm for $\TH$-Reed-Muller codes by reconstructing the error $\TH$-polynomial via majority voting for the unknown coefficients using the corresponding $\GG$-Dickson matrix. We show that this approach permits to correct any error pattern of rank up to half the minimum distance. A detailed complexity analysis of the majority voting algorithm is provided in Appendix \ref{sec:appendix_voting}. Finally, we conclude in Section \ref{sec:conclusion} with a brief summary and some open questions.
\if\longversion1
Finally, in Section \ref{Binary_rankRM}, we consider the binary $\TH$-Reed-Muller codes, i.e., for $\GG \cong (\ZZ/2\ZZ)^m$ for $m \ge 2$ and present a recursive structure which can be considered as an rank analogue of of $(u \vert u+v)$  structure of binary Reed-Muller codes with Hamming metric. We then show that this recursive structure can be exploited to decode binary $\TH$-Reed-Muller codes.
\fi

\section*{Acknowledgements}
The authors express their deep gratitude to the referees for their
careful reading of the article and their very relevant comments that
significantly improved the quality of the paper.

 \section{Preliminaries}\label{sec:2}

\subsection{Notation}
Throughout this paper, $\K$ denotes a field, not necessarily finite,
and $\LL$ denotes a finite Galois extension of $\K$.  We use $\GG$ to
denote the Galois group $\Gal(\LL/\K)$ and the elements of $\GG$ are
usually denoted as $\gg_0, \gg_1, \ldots{}, \gg_{m-1}$. By $\B$, we
denote a basis of the finite dimensional vector space $\LL$ over
$\K$. For a $\K$-linear space $V$, the span of vectors
$\boldsymbol{v}_1, \dots, \boldsymbol{v}_t \in V$ is denoted as $\Span{\boldsymbol{v}_1, \dots,
  \boldsymbol{v}_t}{\K}$. The space of linear endomorphisms of $V$ is denoted
${\rm End}_{\K}(V)$. The space of matrices with $m$ rows and $n$ columns with entries in
$\K$ is denoted using $\K^{m \times n}$ and $\LL^n$ denotes the space
of vectors of length $n$ over $\LL$. Matrices are usually denoted in
bold capital letters, their $(i,j)$-th entry is usually denoted as
$\mathbf{A}_{i,j}$ and we use $\rk(\mathbf{A})$ to denote rank of a
matrix $\mathbf{A}$. We will regularly handle finite sets and sets of indices in the sequel, we introduce the notation $[a,b]$ to denote intervals \textbf{of integers}, namely the finite set $\{a, a+1, \dots, b\}$. Given two subsets $I \subseteq [1,m]$ and
$J \subseteq [1,n]$, we denote $\mathbf{A}_{I,J}$ the submatrix
of $\mathbf{A}$ obtained by keeping only entries with indexes in
$I \times J$.

Finally, when handling complexities we will use Landau notation for comparison.
Namely, for $m$ going to infinity we denote
\begin{align*}
  f(m) = \OO(g(m)) \quad &\text{if}\quad \exists M >0, \quad\text{such that} \quad \forall m \geq M,\ f(m)\leq K g(m)\quad \text{for some }K >0;\\
  f(m) = \Omega(g(m)) \quad &\text{if}\quad \exists M >0, \quad\text{such that} \quad \forall m \geq M,\ f(m)\geq K g(m)\quad \text{for some }K >0;\\
  f(m) = \Theta(g(m)) \quad &\text{if both}\quad f(m) = \OO(g(m)) \quad \text{and}\quad f(m) = \Omega(g(m)).
\end{align*}
Also we denote $f(m) = \widetilde{\OO}(g(m))$ if $f(m)=\OO(g(m)P(\log(m)))$ for some polynomial $P$.

\medskip

In this section, we recall the relevant definitions and basic notions
of rank metric codes as well as their various equivalent
representations. We also record some results about Reed-Muller codes
with rank metric from \cite{ACLN} and derive some properties of
$\GG$-Dickson matrices that will be used in the subsequent sections.
\subsection{Matrix codes}
Delsarte introduced rank metric codes in \cite{Del} as $\K$-linear
subspaces of the matrix space $\K^{m \times n}$ where the rank
distance of two codewords (\emph{i.e.}, matrices) $\mathbf{A}, \mathbf{B} \in \K^{m \times n}$
is given by
\[
d_{\rk}(\mathbf{A},\mathbf{B}) = \rank(\mathbf{A}-\mathbf{B}).
\]
Such matrix spaces are called \emph{matrix rank metric codes} and
denoted by \emph{$[m \times n, k,d]_{\K}$--codes} where $k$ denotes the
$\K$-dimension of the code and $d$ denotes the minimum distance, \emph{i.e.} the minimum of the rank distances of any two distinct codewords.

\begin{rem}
  Note that a more abstract point of view can be adopted by
  considering subspaces of the space of $\K$-linear maps from a finite
  dimensional $\K$-linear space $V$ to another $\K$-linear space
  $W$. This point of view is somehow considered in the sequel when we
  deal with subspaces of skew group algebras (see
  \S~\ref{sec:LG-codes}).
\end{rem}

\subsection{Vector codes}
Since the works of Gabidulin \cite{Gab}, the classical literature
on rank metric codes also involves $\LL$-linear subspaces
of $\LL^n$, where the rank of a vector is defined as
\[
  \rk_{\K}(\aa) \eqdef \dim_{\K} \Span{a_1, \dots, a_n}{\K}.
\]
Next, the distance between two vectors $\aa, \, \bb \in \LL^n$ is
defined as
\[
d_{\rk}(\aa,\bb) \eqdef \rk_{\K}(\aa-\bb).
\]
$\LL$--subspaces of $\LL^n$ are called \emph{vector rank metric codes}
and denoted by $[n, k, d]_{\LL/\K}$ codes where $k$ denotes the
$\LL$--dimension of the code and $d$ denotes the minimum distance.

It is well--known that such vector codes actually can be turned into
matrix codes by choosing a $\K$--basis $\B = ({\beta}_1, \dots, {\beta}_m)$ of
$\LL$ and proceeding as follows. Given an element $x$ of $\LL$ denote by $x^{(1)}, \dots, x^{(m)}$ its coefficients in the basis $\B$. That is to say
$x = x^{(1)}{\beta}_1 + \cdots + x^{(m)} {\beta}_m$ then consider the map
\[
  \text{Exp}_\B : \map{\LL^n}{\K^{m \times n}}{(x_1, \dots, x_n)}{
    \begin{pmatrix}
      x_1^{(1)} & \cdots & x_n^{(1)}\\
      \vdots &  & \vdots \\
      x_1^{(m)} & \cdots & x_n^{(m)}      
    \end{pmatrix}.
}
\]
Then, any vector code $\CC \subset \LL^n$ can be turned into a matrix
code by considering $\text{Exp}_{\B}(\CC)$. The induced matrix code
depends on the choice of the basis $\B$ but choosing another basis provides
an isometric code with respect to the rank metric.

\begin{rem}
  Note that if an $\LL$--linear rank metric code can be turned into a
  matrix code, the converse is not true. A subspace of
  $\K^{m\times n}$ can be turned into a $\K$-linear subspace of $\LL^n$
  by applying the inverse map of $\text{Exp}_\B$ but the resulting
  code will not be $\LL$--linear in general. Thus, codes of the form
  $\text{Exp}_\B(\CC)$ when $\CC$ ranges over all $\LL$--subspaces of
  $\LL^n$ form a proper subclass of matrix codes in $\K^{m\times n}$.
\end{rem}

\subsection{Rank metric codes as
  \texorpdfstring{$\LL[\GG]$}{}--codes}\label{sec:LG-codes}

The study of rank metric codes as $\LL$-subspaces of the skew group algebra
$\LL[\GG]$ has been initiated in \cite{ACLN}. It generalizes the
study of rank metric codes over arbitrary cyclic Galois extensions in \cite{ALR18}. We recall below the definitions and
basic notions of rank metric codes in this setting.

Consider an arbitrary but fixed finite Galois extension $\LL/\K$ with
$\GG \eqdef \Gal(\LL/\K)$. The \emph{skew group} algebra $\LL[\GG]$ of $\GG$ over $\LL$ is defined as
\[
\LL[\GG] : = \Bigg\{\sum_{\gg \in \GG} a_\gg \gg ~\colon~ a_\gg \in \LL\Bigg\}
\]
and endowed with its additive group structure and the following composition law derived from the group law of $\GG$:
\[
(a_\gg \gg ) \circ (a_\hg \hg) = (a_\gg \gg(a_\hg)) (\gg \hg),
\]
which is extended by associativity and distributivity.
This equips $\LL[\GG]$ with a non-commutative algebra structure.

\begin{thm}\label{thm:correspondence_poly_endo}
  Any element $A = \sum_\gg a_\gg \gg \in \LL[\GG]$ defines a
  $\K$-endomorphism of $\LL$ that sends $x \in \LL$ to
  $\sum_\gg{a_\gg\gg(x)}$. This correspondence induces
  a $\K$-linear isomorphism between $\LL[\GG]$ and
  $\ensuremath{\textrm{End}}_{\K}(\LL)$.
\end{thm}

\begin{proof}
  See for instance {\cite[Thm.~1]{ACLN}}.
\end{proof}
Thus, the \emph{rank} of an element $A \in \LL[\GG]$ is well-defined as its rank when viewed as a $\K$-linear endomorphism of $\LL$.
From the above theorem, it is clear that with respect to a fixed basis
$\B = ({\beta}_1, \dots, {\beta}_m)$ of $\LL/\K$, we get $\K$--linear isomorphisms
\begin{equation}\label{representation}
  \LL[\GG] \cong \textrm{End}_\K(\LL) \cong \K^{m\times m}.
\end{equation}
Also, w.r.t. the basis $\B = \{ {\beta}_1, \ldots, {\beta}_m\}$, every element  $A \in \LL[\GG]$ can be seen as a vector
\[
  \textbf{a} = (A({\beta}_1), \ldots, A({\beta}_m)) \in \LL^m.
\]

Moreover, the aforementioned definition of rank for a vector of $\LL^m$
coincides with the rank of $A$ when regarded as a $\K$--endomorphism
of $\LL$ (according to Theorem~\ref{thm:correspondence_poly_endo}).

\begin{rem}
In the particular case of a Galois extension of finite fields
$\LL/\K$, the group $\GG$ is cyclic, say, $\GG = \< \sigma \>$ and there
are many characterizations of $\LL[\GG]$ studied in \cite{WL13}. One of
the very well--known characterization is in terms of \emph{linear
polynomials} studied by Ore \cite{Oreqpoly} followed by his work on the
theory of non-commutative polynomials \cite{Orenoncom}.  Let
$\K = \Fq$ and $\LL = \Fqm$ for some prime power $q$ and a positive
integer $m$, then the \emph{linear polynomials} over $\Fqm$ are given by
\[
  L(x) = \sum\limits_{i=0}^{d} a_ix^{q^i},\ \  \text{for\ some\ } d \in \NN \quad
  \text{and}\quad a_0, \dots, a_d \in \Fqm
\]
and endowed with the composition law to give a structure of (non
commutative) ring. With this point of view, the skew group algebra
$\Fqm[\GG]$ is isomorphic to the ring of linear polynomials modulo
the two--sided ideal generated by $x^{q^m}-x$.
\end{rem}

\begin{defn}
    An $\LL$-linear rank metric code $\CC$ in the skew group algebra $\LL[\GG]$ is an $\LL$-linear subspace of $\LL[\GG]$, equipped with the rank distance. The dimension of $\CC$ is defined as its dimension as $\LL$-vector space. The minimum rank distance is defined as 
    \[
    d(\CC) \eqdef \min \{\rk(A) ~\colon~ A \in \CC\setminus \{0\}\},
  \]
  where the rank of $A \in \LL [\GG]$ is the rank of the
  $\K$--endomorphism it induces on $\LL$ (according to
  Theorem~\ref{thm:correspondence_poly_endo}).
  \end{defn}
  
  We denote the parameters of an $\LL$-linear rank metric code
  $\CC \subseteq \LL[\GG]$ of dimension $k$ and minimum distance $d$ by
  $[m,k,d]_{\LL [\GG]}$ where $m$ denotes the extension degree
  $[\LL : \K]$. If $d$ is unknown or clear from the context, we
  simply write $[m,k]_{\LL[\GG]}$.

As observed earlier, an element of $\LL[\GG]$ can be seen as a
  $\K$-linear endomorphism of $\LL$. Therefore, if we fix a $\K$-basis
  $\B$ of $\LL$, then, after suitable choices of bases, one can
  transform an $[m,k,d]_{\LL[\GG]}$--code $\CC$ into an
  $[m \times m, k, d]_{\K}$--code or into an $[m,k,d]_{\LL/\K}$ code.

\subsection{Rank metric Reed-Muller codes}\label{sec:Rank_RM}
Introduced in \cite{ACLN}, rank metric Reed--Muller codes are subcodes
of the skew group algebra of a finite extension whose Galois group is
a product of cyclic groups.
Let $\GG$ be the product of cyclic groups
$\ZZ/{n_1\ZZ} \times \cdots \times \ZZ/{n_m \ZZ}$. For the skew group
algebra representation, we need a multiplicative description of the
group. For this sake, we introduce a system of generators:
$\theta_1, \ldots, \theta_m$ so that
$\theta_1^{i_1}\cdots \theta_m^{i_m}$ describes the $m$--tuple
$(i_1, \dots, i_m) \in \ZZ/n_1 \ZZ \times \cdots \times \ZZ/n_m\ZZ$.

Let $\mathbf{n}$ denote the tuple $(n_1, \dots, n_m)$. Denote
\begin{equation}\label{eq:LambdaTheta}
  \Lambda(\mathbf{n}) \eqdef [0,n_1-1] \times \cdots \times [0,n_m-1] \quad
  \text{and} \quad
  \TH^{\mathbf{i}} \eqdef \theta_1^{i_1} \cdots
  \th_m^{i_m} \in \GG \quad \text{for} \quad
  \mathbf{i} = (i_1,\dots, i_m) \in \mathbb{N}^m.
\end{equation}
Then $\GG = \{\TH^{\mathbf{i}} ~\colon~ \mathbf{i} \in \Lambda(\mathbf{n})\}$ and hence any $P \in \LL[\GG]$ has a unique representation  
    \[
    P = \sum\limits_{\mathbf{i} \in \Lambda(\mathbf{n})} b_{\mathbf{i}}\TH^{\mathbf{i}}.
    \]
Because of the above description, elements $P \in \LL [\GG]$ are referred
    to as \emph{$\TH$--polynomials}.
    
    \begin{defn}\label{thpoly}
      The \emph{$\TH$-degree} of a $\TH$--monomial
      $\TH^{\mathbf{i}} = \theta_1^{i_1}\cdots \theta_m^{i_m}$ where
      $\mathbf i \in \Lambda(\mathbf{n})$ is $i_1+\cdots + i_m$.  The
      \emph{$\TH$--degree} of a $\TH$-polynomial $P\in \LL [\GG]$,
      denoted $\deg_{\TH} P$ is the maximal degree of its monomials.
    \end{defn}

A rank analogue of Reed-Muller code is defined as follows.
    \begin{defn}[{\cite[Def.~8.1]{ACLN}}]\label{theta_RM}
       Let $r \in \NN$ such that $r \le \sum_i (n_i -1)$.
        The $\TH$-Reed-Muller code of order $r$ and type $\nn$ is
        \[
        {\RM}_{\TH}(r, \nn) \eqdef \{P \in \LL[\GG] ~\colon~ \deg_{\TH} P \le r \}.
        \]
    \end{defn}

    \begin{rem}
      As mentioned in \cite[Rem.~45]{ACLN}, the definitions of
      $\TH$--degree and $\TH$--Reed-Muller codes depend on the choice
      of the generators of $\GG$. Note that even in the setting of
      cyclic extensions, different choices of generators provide
      either Gabidulin codes or generalized Gabidulin codes.
    \end{rem}

    With respect to a fixed basis $\B=({\beta}_1, \dots, {\beta}_m)$ of
    the finite Galois extension $\LL/\K$, the $\TH$-Reed-Muller code
    can be seen as vector code as
\[
  \{(P({\beta}_1), \dots, P({\beta}_m)) ~:~ P \in \LL[\GG], \, \deg_{\TH}(P) \leq r\} \subseteq
  \LL^{N}.
\]
Finally, the exact parameters of these codes are known.

\begin{thm}[{\cite[Prop.~48 \& Thm.~50]{ACLN}}]\label{min}
  For $\mathbf{n} = (n_1 \ge n_2 \ge \cdots \ge n_m\ge 2)$, let $s $ and $\ell$ be the unique integers such that
  $r = \sum_{i=s+1}^{m} (n_i - 1) + \ell$ with
  $0 \leq \ell < n_s$. Then the code $\RM_{\TH}(r, \nn)$ has
  dimension $|\{ \mathbf{i} \in \Lambda(\nn)\colon |\mathbf{i}| \le r\}|$ and minimum distance
\[ d = (n_s - \ell) \prod\limits_{i=1}^{s-1}n_i.\]
\end{thm}

\begin{exa}\label{runningexample}
  Let us fix $\K = \Q$, and $\LL$ be the splitting field of the
  polynomial $(x^2 - 2)(x^2 -3) (x^2 -5)$. Therefore,
  $\LL = \Q(\sqrt{2}, \sqrt{3}, \sqrt{5})$ and the Galois group
  $\GG = \Gal(\LL/\K)$ is isomorphic to the Abelian group
  $(\ZZ/{2\ZZ})^3$, which is generated by the automorphisms $\th_i$,
  for $i = 1, 2, 3$, defined as
\[
\th_1 \colon 
\begin{cases}
\sqrt{2} \mapsto - \sqrt{2}\\
\sqrt{3} \mapsto \sqrt{3}\\
\sqrt{5} \mapsto \sqrt{5}
\end{cases} \quad \text{and} \quad 
\th_2 \colon
\begin{cases}
\sqrt{2} \mapsto  \sqrt{2}\\
\sqrt{3} \mapsto - \sqrt{3}\\
\sqrt{5} \mapsto \sqrt{5}
\end{cases}
\quad \text{and} \quad 
\th_3 \colon
\begin{cases}
\sqrt{2} \mapsto  \sqrt{2}\\
\sqrt{3} \mapsto \sqrt{3}\\
\sqrt{5} \mapsto - \sqrt{5}.
\end{cases}
\]
Consider the rank metric code $\CC = \RM_{\TH}(1,(2,2,2))$ given by 
\begin{equation}\label{exa}
\CC\eqdef \{ a \cdot \Id + b \cdot \th_1 + c \cdot \th_2 + d \cdot \th_3 ~\colon~ a, b, c, d \in \LL\}.
\end{equation}
According to Theorem~\ref{min} this code is $[8,4,4]_{\LL/\K}$.  We
fix the following ordered basis
\[\B = \{1, \sqrt{2}, \sqrt{3}, \sqrt{6}, \sqrt{5}, \sqrt{10},
  \sqrt{15}, \sqrt{30} \}\] of $\LL/\K$. Then, the
$[8,4]_{\LL/\K}$ code $\CC(\B)$ is generated by the $4 \times 8$ matrix
\[
\begin{pmatrix}
1 & \sqrt{2} & \sqrt{3} &  \sqrt{6} & \sqrt{5} & \sqrt{10} & \sqrt{15} &\sqrt{30} \\
1 & - \sqrt{2} & \sqrt{3} & - \sqrt{6}&\sqrt{5}  & - \sqrt{10} & \sqrt{15} & - \sqrt{30}\\
1 & \sqrt{2} & - \sqrt{3} & - \sqrt{6} & \sqrt{5} & \sqrt{10} & - \sqrt{15} & - \sqrt{30} \\
1 & \sqrt{2} & \sqrt{3}  &\sqrt{6}  & - \sqrt{5}  & - \sqrt{10} &  - \sqrt{15}  & - \sqrt{30} 
\end{pmatrix}.
\]

We can also represent the codewords as $8 \times 8$ matrices over $\K$
which are the coordinate matrices w.r.t. the basis $\B$. The matrices
that represent the multiplication by the elements of the basis $\B$
are of the form $A^i B^j C^k$ for $i, j, k \in \{ 0,1\}$, where
\[\footnotesize\setlength\arraycolsep{3pt}
A = \begin{pmatrix} 
0 &  2 &  0 & 0 & 0 & 0 & 0 & 0\\
1 &  0 &  0 & 0 & 0 & 0 & 0 & 0\\
0 &  0 &  0 & 2 & 0 & 0 & 0 & 0\\
0 &  0 &  1 & 0 & 0 & 0 & 0 & 0\\
0 &  0 &  0 & 0 & 0 & 2 & 0 & 0\\
0 &  0 &  0 & 0 & 1 & 0 & 0 & 0\\
0 &  0 &  0 & 0 & 0 & 0 & 0 & 2\\
0 &  0 &  0 & 0 & 0 & 0 & 1 & 0
 \end{pmatrix},  
\, B =
\begin{pmatrix}
0 &  0 &  3 & 0 & 0 & 0 & 0 & 0\\
0 &  0 &  0 & 3 & 0 & 0 & 0 & 0\\
1 &  0 &  0 & 0 & 0 & 0 & 0 & 0\\
0 &  1 &  0 & 0 & 0 & 0 & 0 & 0\\
0 &  0 &  0 & 0 & 0 & 0 & 3 & 0\\
0 &  0 &  0 & 0 & 0 & 0 & 0 & 3\\
0 &  0 &  0 & 0 & 1 & 0 & 0 & 0\\
0 &  0 &  0 & 0 & 0 & 1 & 0 & 0
\end{pmatrix},
\, C = 
\begin{pmatrix}
0 &  0 &  0 & 0 & 5 & 0 & 0 & 0\\
0 &  0 &  0 & 0 & 0 & 5 & 0 & 0\\
0 &  0 &  0 & 0 & 0 & 0 & 5 & 0\\
0 &  0 &  0 & 0 & 0 & 0 & 0 & 5\\
1 &  0 &  0 & 0 & 0 & 0 & 0 & 0\\
0 &  1 &  0 & 0 & 0 & 0 & 0 & 0\\
0 &  0 &  1 & 0 & 0 & 0 & 0 & 0\\
0 &  0 &  0 & 1 & 0 & 0 & 0 & 0
 \end{pmatrix}
\]
represent the multiplication by $\sqrt{2}$, $\sqrt{3}$, and $\sqrt{5}$, respectively. 

The matrix representation of the code is the $\Q$-span of the set
\[\{A^i B^j C^k, A^i B^j C^k X, A^i B^j C^k Y, A^i B^j C^k Z ~\colon~ 0
\le i, j, k \le 1 \},\] where
\[\footnotesize\setlength\arraycolsep{3pt}
X = \begin{pmatrix}
1 &  0  &  0 &  0 &  0 &  0  & 0 & 0\\
0 & -1  &  0 &  0 &  0 &  0  & 0 & 0\\
0 &  0  &  1 &  0 &  0 &  0  & 0 & 0\\
0 &  0  &  0 & -1 &  0 &  0  & 0 & 0\\
0 &  0  &  0 &  0 &  1 &  0  & 0 & 0\\
0 &  0  &  0 &  0 &  0 & -1  & 0 & 0\\
0 &  0  &  0 &  0 &  0 &  0  & 1 & 0\\
0 &  0  &  0 &  0 &  0 &  0  & 0 & -1
\end{pmatrix},
\, \quad
Y = \begin{pmatrix}
1 &  0  &  0 &  0 &  0 &  0  & 0 & 0\\
0 &  1  &  0 &  0 &  0 &  0  & 0 & 0\\
0 &  0  & -1 &  0 &  0 &  0  & 0 & 0\\
0 &  0  &  0 & -1 &  0 &  0  & 0 & 0\\
0 &  0  &  0 &  0 &  1 &  0  & 0 & 0\\
0 &  0  &  0 &  0 &  0 &  1  & 0 & 0\\
0 &  0  &  0 &  0 &  0 &  0  &-1 & 0\\
0 &  0  &  0 &  0 &  0 &  0  & 0 & -1
\end{pmatrix},
\quad
\footnotesize\setlength\arraycolsep{3pt}
Z= \begin{pmatrix}
1 &  0  &  0 &  0 &  0 &  0  & 0 & 0\\
0 &  1  &  0 &  0 &  0 &  0  & 0 & 0\\
0 &  0  &  1 &  0 &  0 &  0  & 0 & 0\\
0 &  0  &  0 &  1 &  0 &  0  & 0 & 0\\
0 &  0  &  0 &  0 & -1 &  0  & 0 & 0\\
0 &  0  &  0 &  0 &  0 & -1  & 0 & 0\\
0 &  0  &  0 &  0 &  0 &  0  &-1 & 0\\
0 &  0  &  0 &  0 &  0 &  0  & 0 & -1
\end{pmatrix}
\]
represent the matrices $\th_1$, $\th_2$ and $\th_3$ in the basis $\B$, respectively.
\end{exa}

\subsection{Dickson matrices}
A crucial tool in the sequel consists in identifying the elements of
the skew group algebra with their corresponding \emph{$\GG$-Dickson matrices}
that we define below. For defining $\GG$-Dickson matrices, first we fix
an ordering $\{\gg_0, \ldots, \gg_{m-1}\}$ on $\GG$. Also, let
$\sigma_i \in \mathfrak{S}_m$ be the permutation representation of $\gg_i$
induced by the left action of $\GG$ onto itself, \emph{i.e.},
\[\sigma_i(j) = k \quad \text{if}  \quad \gg_i\gg_j = \gg_k \text{ for } j \in \{0, \ldots, {m-1}\}.\]

\begin{defn}[{\cite[Def.~14]{ACLN}}]\label{Dicksondefn}
  The $\GG$-Dickson matrix associated to
  $A= \sum_{i=0}^{m-1} a_i \gg_i \in \LL[\GG]$ is the matrix
  representing the following $\LL$--linear map in the basis
  $(\gg_0, \dots, \gg_{m-1})$:
  \[
    \map{\LL[\GG]}{\text{End}_{\LL}(\LL[\GG])}{A}{(B \mapsto B \circ A).}    
    \]
    It is defined as
  $\Dick_{\GG}(A) \eqdef {(d_{i,j})}_{i,j} \in \LL^{m \times m}$ where
    \(d_{i,j} = \gg_j\left(a_{\sigma_{j}^{-1}(i)}\right)\).
    \end{defn}
    The $\GG$-Dickson matrices are indeed the usual Dickson matrices for an extension of finite fields $\LL/\K$ as we record in the example below.

\begin{exa}
  For
  $\LL = \F_{q^m}, \, \text{ let } \GG = \Gal(\F_{q^m}/\F_q) = \< \th
  \>,$ where $\th$ is the Frobenius map. Then w.r.t. the ordered basis
  $(Id, \th, \dots, \th^{m-1})$ of $\LL[\GG]$, the $\GG$-Dickson matrix
  of $F = \sum_{i=0}^{m-1} f_i x^{[i]}$ is
    \begin{equation}\label{finiteDickson} \Dick_\GG(F) =\begin{pmatrix}
    f_0 &f_{m-1}^q &\cdots &f_1^{q^{m-1}}\\
    f_1 &f_0^{q} &\cdots &f_2^{q^{m-1}}\\
    \vdots & &\ddots &\\
    f_{m-1} & f_{m-2}^q &\cdots &f_0^{q^{m-1}}
 \end{pmatrix}.
 \end{equation}

\end{exa}

One of many equivalent ways of determining rank of an element of $\LL[\GG]$ is by rank of its $\GG$-Dickson matrix. This is a generalization of the finite field case (see, \emph{e.g.}, \cite{WL13}).

\begin{pro}[{\cite[Thm.~21 \& Thm.~24]{ACLN}}]\label{prop:Dickson_rank}
The algebra
\[
  \mathcal{D} (\LL /\K ) \eqdef \left\{ \Dick_G(A)^\top  ~:~  A \in  \LL[\GG]\right\}  \subseteq  \LL^{m \times m}
\]
is isomorphic to $\LL[\GG]$. Moreover, for any $A \in \LL [\GG]$, we
have $\rk(A) = \rank (\Dick_{G}(A))$.
\end{pro}

A rank preserving representation of a vector of $\LL^N$ is given by its associated $\GG$-Moore matrix, analogous to the Moore/Wronskian matrix in the finite field case.

\begin{defn}[{\cite[Def.~7]{ACLN}}]
For a vector $\vv =(v_1, \ldots, v_m) \in \LL^m$, its $\GG$-Moore matrix is defined as
\[
M_{\GG}(\vv)\eqdef
\begin{pmatrix}
\gg_0(v_1) & \ldots & \gg_0(v_m)\\
\vdots &\ddots &\vdots \\
\gg_{m-1}(v_1) &\ldots &\gg_{m-1}(v_m)
\end{pmatrix}.
\]
\end{defn}

It is proved in \cite[Prop.~9]{ACLN} that $\rk_{\K}(\vv) = \rk_{\LL}(M_{\GG}(\vv))$. Abusing the notation, we will also use $M_{\GG}(\tilde{\vv})$ to denote the truncated $\GG$-Moore $m \times p$ matrix in the case $\tilde{\vv} \in {\LL}^p$ where $p<m$.

Next, we give a decomposition of a $\GG$-Dickson matrix
associated to $A \in \LL[\GG]$ into product of two truncated $\GG$-Moore matrices based on a trace representation of $A$. The representation of a linear polynomial of rank $k$ is essentially proved in \cite[Thm.~2.4]{LQ12}. We give a proof for an arbitrary finite
Galois extension for completeness. 
\begin{pro}\label{Dickson}
  Let $\LL/\K$ be a finite Galois extension of degree $m$ with Galois group $\GG$. If an element $A = \sum_{\gg\ \!\! \in \GG} a_\gg \gg \in \LL[\GG]$ has rank
  $t$, then there exist two vectors $\av = (\a_1, \ldots, \a_t )$ and $\bb = (\b_1, \ldots, \b_t ) \in \LL^t$, with $\rk_{\K}(\av) = \rk_{\K}(\bb) =t$, such that $\Dick_{\GG}(A) = M_{\GG}(\bb) M_{\GG}(\av)^{\top}$, where $M_{\GG}(\av), \, M_{\GG}(\bb)$ are the truncated $\GG$-Moore matrices of order $m \times t$.
\end{pro}

\begin{proof}
  First, we prove that there exist subsets of two $\K$-linearly independent
  elements
       $\{\a_1, \ldots, \a_t \}$, and
       $\{\b_1, \ldots, \b_t \} \subseteq \LL$ such that $A$, when regarded
       as a $\K$--endomorphism of $\LL$ satisfies:
     \begin{equation}\label{Traceform}
     A = \sum_{i=1}^{t} \a_i T_{\b_i},
     \end{equation}
     where $T_{\b_i}$ is the $\K$--homomorphism from $\LL$ to $\K$ defined as
     $T_{\b_i}(x) = \text{Tr}_{\LL/\K}(\b_ix) = \sum_{\gg \in \GG}
     \gg(\b_i x)$.
     To see this, let $(b_{t+1}, \dots, b_{m})$ be a
     $\K$--basis of $\ker(A)$ that we complete into a basis
     $(b_1, \dots, b_m)$ of $\LL$.  Let $(\beta_1, \dots , \beta_m)$
     be the dual basis of $(b_1, \dots, b_m)$ with respect to the
     bilinear form $(x,y) \mapsto \text{Tr}_{\LL/\K}(xy)$.
     Then,
     \[
       A = A(b_1)\text{Tr}_{\LL/K}(\beta_1 x) + \cdots + A(b_t)
       \text{Tr}_{\LL/\K}(\beta_t x).
     \]
     Indeed, the right hand side
     evaluates like $A$ at $b_1, \dots, b_m$.
     Finally, since $b_1, \dots, b_t$ span a complement subspace of $\ker A$,
     the elements $A(b_1), \dots, A(b_t)$ are linearly independent,
which proves \eqref{Traceform}.  

Now, note that, for any $k\in [1, m]$, $T_{\beta_k}(x) = \sum_i \gg_i(\beta_k) \gg_i (x)$. Hence,
regarded as an element of $\LL[\GG]$ it equals to $\sum_i \gg_i(\beta_k)\gg_i$. Thus, \eqref{Traceform} entails
\begin{align*}
  A = \sum_{k=1}^{t} \a_k T_{\b_k} &= \sum_{k=1}^{t} \a_k \sum_{i=0}^{m-1} \gg_i(\b_k)\gg_i\\
                                          &= \sum_{i=0}^{m-1} \sum_{k=1}^{t} \a_k \gg_i(\b_k)\gg_i,
\end{align*}
and therefore the coefficients of
$A = \sum_{i=0}^{m-1} a_i \gg_i$ satisfy
\[
  \forall\ \!  0 \leq i
  \leq m-1,\quad
  a_i = \sum_{k=1}^{t} \a_k \gg_i(\b_k).\]
According to Definition~\ref{Dicksondefn}, 
the $(i,j)$-th entry of $\Dick_\GG(A)$ is
$\gg_j(a_{\sigma_j^{-1}(i)}) = \gg_j(\sum_k \a_k
\gg_{\sigma_j^{-1}(i)}(\b_k))$
$= \sum_k \gg_j(\a_k) \gg_i(\b_k)$. Therefore,
  \begin{equation}\label{DicksonasMoore}
    \Dick_\GG(A) = \begin{pmatrix}
        \gg_0(\b_1) & \ldots & \gg_0(\b_t)\\
        \vdots & \ddots & \vdots\\
        \gg_{m-1}(\b_1) & \ldots & \gg_{m-1}(\b_t)
    \end{pmatrix}
    \begin{pmatrix}
        \gg_0(\a_1) & \ldots &\gg_{m-1}(\a_1)\\
        \vdots & \ddots &\vdots\\
        \gg_0(\a_t) &\ldots &\gg_{m-1}(\a_t)
    \end{pmatrix}.
    \end{equation}

\end{proof}
 The following corollary gives a very important property of $\GG$-Dickson matrices when $\GG$ is cyclic and will be useful for the decoding algorithms to follow. The result for the finite field case was proved in \cite[Thm.~3]{Tovo} and was used for decoding of Gabidulin codes over finite field extensions.
\begin{cor}\label{consecutive}
Let $\LL/\K$ be a cyclic Galois extension with Galois group $\GG = \< \theta\>$. If the elements of $\GG$ are ordered as $\gg_i = \theta^i$ for $i \in [0, |\GG|-1]$, then any  $t \times t$ submatrix of the $\GG$-Dickson matrix $\Dick_{\GG}(A)$ of an element $A = \sum_{\gg\ \!\! \in \GG} a_\gg \gg \in \LL[\GG]$ formed by $t$ consecutive rows and $t$ consecutive columns is invertible.
\end{cor}

\begin{proof}
Following the decomposition in \eqref{DicksonasMoore} if we write 
    $\Dick_\GG(A) = \mathbf{M}_1 \mathbf{M}_2^{\top}$, then any $t \times t$ submatrix of $\Dick_\GG(A)$ formed by $t$ consecutive columns and rows is obtained by product of a submatrix of $\mathbf{M}_1$ of $t$ consecutive rows with a submatrix of $\mathbf{M}_2^{\top}$ of $t$ consecutive columns. {It is clear that the matrices $\mathbf{M}_1$, $\mathbf{M}_2$ are truncated $\GG$-Moore matrices $ M_{\GG}(\bb), \, M_{\GG}(\av)$ respectively, where $ \bb = (\beta_1, \dots, \beta_t), \, \av = (\alpha_1, \dots, \alpha_t)\in  \LL^{t}$.}
As both $(\alpha_1, \dots, \alpha_t)$
and $(\beta_1, \dots, \beta_t)$ are linearly independent, it follows from \cite[Prop.~9]{ACLN} that $\mathbf{M}_1,\, \mathbf{M}_2$ have full rank. Now we show that any $t \times t$ submatrix of $\mathbf{M}_1$ (resp.$\mathbf{M}_2$) consist of $t$ consecutive rows are invertible. Indeed, if any $t$ consecutive rows of $\mathbf{M}_1$ are $\K$-linearly dependent, so will be the first $t$ consecutive rows due to our choice of the ordering on the elements of $\GG$.
Suppose that for some $1 < t_0 \leq t$, the $t_0$--th row is an $\LL$--linear combination of the $t_0-1$ previous ones.
Then, iteratively applying $\theta$ to the rows we deduce that any row is a linear combination of the $t_0-1$ previous ones and hence that the row
  space of $\mathbf{M}_1$ is generated by the first $t_0-1$ rows which contradicts the rank of $\mathbf{M}_1$.
Hence, it completes the proof.
\end{proof}
\begin{rem}
Whether the statement in Corollary \ref{consecutive} still holds for arbitrary Abelian group $\GG$ is still unclear. But it should be noted that cyclicity of $\GG$ is not assumed in getting the decomposition of the $\GG$-Dickson matrix into product of two truncated Moore matrices as shown in \eqref{DicksonasMoore}. Thus, whether Corollary \ref{consecutive} is true for arbitrary $\GG$-Dickson matrices or not depends on whether $t$ consecutive rows of these truncated Moore matrices defined over arbitrary Abelian groups are invertible or not. This does not seem to be true in general.
For instance, let $\gg_0 = Id, \gg_1 = \theta_1$ for the extension $\mathbb{Q}(\sqrt{2}, \sqrt{3}, \sqrt{5})/\mathbb{Q}$ of Example~\ref{runningexample}. If we take the vector $\vv = (1, \sqrt{3})$, then then first two rows of the truncated Moore matrix are linearly dependent (in fact, same) as $\theta_1$ fixes $\sqrt{3}$.
\end{rem}

 \section{Decoding using $\GG$-Dickson matrices}\label{sec:dickson_framework}

In the sequel, we always fit in the following context. Let $\LL/\K$
be a finite Galois extension with Galois group $\GG$. Suppose
$\CC \subseteq \LL[\GG]$ is a rank metric Reed--Muller code
$RM_{\TH}(r, \nn)$ with minimum rank distance $d$ and we are given
\[
  Y = C+E,
\]
where $C \in \CC$ and $E \in \LL[\GG]$ with $\rk (E) = t \le \lfloor \frac{d-1}{2}\rfloor$.

The $\GG$-Dickson matrix based decoding of rank metric Reed--Muller code $RM_{\TH}(r, \nn)$ can be seen as an instance of the problem of reconstruction of $\TH$-polynomials. Indeed, by denoting 
\[
Y = \sum_{\gg \in \GG} y_\gg \gg, \quad
C = \sum_{\gg \in \GG} c_\gg \gg \quad \text{and} \quad
E = \sum_{\gg \in \GG} e_\gg \gg,
\]
the primary observation one can make is the following.
\begin{description}
\item[\textbf{Observation. $E$ is partially known}] The element $Y$
  is known and we aim to compute the pair $(C,E)$.  Since
  $C \in \RM_{\TH}(r,\nn)$ and $Y = C+E$, then for any $\gg \in \GG$
  with $\TH$-degree $>r$ we have $c_\gg = 0$ and hence
  $y_\gg = e_\gg$. In summary: \textbf{for any $\gg$ of $\TH$--degree $>r$,
  $e_\gg$ is known}. Therefore, {$\Dick_\GG(E)$ is partially
  known}.
  \end{description}
  
 \medskip
 
We will reconstruct the error $\TH$-polynomial $E$ by recovering its unknown coefficients as follows.
\begin{description}
\item[\textbf{Main idea.  One may iteratively compute the unknown coefficients of $E$}]

The strategy is to find submatrices of the $\GG$-Dickson matrix $\Dick_\GG (E)$ that contain only one unknown entry denoted as $x$ such that the row containing $x$ can be written as a linear combination of the rest of the rows. 
\[
\begin{pmatrix}
(*) & \dots & x \\
(*) & \dots & (*)\\
    &\ddots&    \\
(*) & \dots & (*)    
\end{pmatrix}.
\]
\end{description}

\begin{rem}\label{rem:conjugate_unknown}
  As we will see in the sequel, the unknown entry $x$ is in general not
  exactly a coefficient $e_\gg$ of $E$ but its image $\hh(e_\gg)$ by some
  $\hh \in \GG$ for $\gg, \hh \in \GG$ that are determined by the indexes
  of the unknown entry (see Definition~\ref{Dicksondefn}). Hence $e_\gg$
  can then be recovered from $x$ by applying $\hh^{-1}$.
\end{rem}

Before describing the technique for obtaining a sequence of such submatrices for decoding $\TH$-Reed--Muller codes, we mention an approach for decoding Gabidulin codes over an arbitrary cyclic Galois extension, \emph{i.e.}, when {$\GG$ is cyclic.} In this case, the existence of such a submatrix is proved by finding a $(t+1)\times (t+1)$ submatrix made of consecutive rows and consecutive columns that contains only one unknown entry denoted as $x$:
\[
\begin{pmatrix}
(*) & \vdots& (*) \\
\cdots    &x& \cdots    \\
(*) & \vdots & (*)    
\end{pmatrix}.
\]
 Recall that
  $\rk (E) = \rk(\Dick_\GG (E)) = t$. Therefore, any $(t+1)\times (t+1)$ minor of $\Dick_\GG (E)$ vanishes. 
The determinant of this submatrix vanishes and
expresses as $ax + b$ where $a$ is a $t \times t$ minor of
$\Dick_\GG(E)$, which, from Proposition~\ref{Dickson} is nonzero.
Since $a, b$ do not depend on $x$, they can be computed from
known coefficients. Then, $x$ can be recovered as the unique solution
of the degree $1$ equation given by the cancellation of the
determinant of the above submatrix. For Gabidulin codes, we explain the method of identifying a sequence of $(t+1)\times (t+1)$--submatrices of $\Dick_\GG(E)$ containing only one unknown coefficient in \S \ref{Gcyclic}.

 However, for decoding rank metric Reed--Muller codes  $RM_{\TH}(r, \nn)$, the absence of the property in Corollary \ref{consecutive} required to combine the previously sketched approach with a majority voting technique to recover the unknown coefficients iteratively. 
We discuss the decoding of rank metric Reed--Muller codes in detail in \S~\ref{sec:Dickson_decoding}. This approach is inspired by, though essentially different from, the decoding method using majority voting for unknown syndromes first introduced by Feng and Rao in \cite{FR} for decoding algebraic geometry codes.

 \section{Decoding using Dickson matrices: first examples}\label{sec:first_examples}
In this section, we recall how the property of circulant Dickson matrix stated in Corollary \ref{consecutive} enables to decode Gabidulin codes. Later on, we show how a similar method can be adapted for decoding Reed-Solomon codes, the Hamming counterpart of Gabidulin codes.

\subsection{Illustration: using Dickson matrices to decode Gabidulin codes}

In what follows we show how to use Dickson matrices to decode
a Gabidulin code (\emph{i.e.}, when $\GG$ is cyclic) of dimension $k$.
Somehow, for Gabidulin codes, it consists in adapting the idea
from \cite{Tovo} by using Dickson matrices.

Here, Gabidulin codes are regarded as an $\Fqm$--subspace of the
twisted group algebra $\Fqm[\GG]$ where $\GG$ is the cyclic group of
order $m$ generated by the Frobenius automorphism $\theta$.  In our
setting, the Gabidulin code of dimension $k$ is the
following $\Fqm$--subspace of $\Fqm[\GG]$
\[
\mathcal{G}_k \eqdef {\langle \theta^i ~:~ 0 \leq i < k \rangle}_{\Fqm}.
\]
Equivalently, they correspond to $\TH$--Reed-Muller codes of degree $k-1$
in $\Fqm[\GG]$ (see Definition~\ref{theta_RM}).

\begin{rem}
  Usually in the literature, Gabidulin codes are given in vector
  representation. The conversion from a subspace of $\Fqm [\GG]$ to
  a subspace of $\Fqm^m$ is explained in Section~\ref{sec:LG-codes}.
  Due to this equivalence, our description of Gabidulin code is equivalent
  to that of usual (vector) Gabidulin codes over $\Fqm$ and of length $m$.
\end{rem}

We will show the technique works for
$t = \lfloor \frac{d-1}{2}\rfloor = \lfloor \frac{m-k}{2}\rfloor$. In
this setting, one can observe that the indexes of the involved
$(\lfloor \frac{d-1}{2}\rfloor +1)\times (\lfloor \frac{d-1}{2}\rfloor
+1)$ submatrices can be chosen independently from the error
$E$. Therefore, if the rank $t$ of $E$ turns out to be less than
$\lfloor \frac{d-1}{2}\rfloor$, then decoding remains possible by
considering $(t+1)\times (t+1)$ submatrices of the aforementioned
$(\lfloor \frac{d-1}{2}\rfloor +1)\times (\lfloor \frac{d-1}{2}\rfloor
+1)$ submatrices.  In summary, the decoding process we describe for an
error of rank $\frac{d-1}{2}$ actually easily adapts to errors of
lower ranks.  For this reason, in this section, when describing the
algorithm, we will always assume that
  \[
    t \eqdef \rk (E) = \left\lfloor\frac{d-1}{2}\right\rfloor . 
  \]
 
\subsubsection{A first example}\label{sec:Gabi_1st_example}
To begin, we illustrate the decoding method for a Gabidulin code with
$m = 7$ and $k = 3$.  Suppose the sent message is
$C = c_0 X + c_1 X^q + c_2 X^{q^2}$ and the received message is
$Y =C+ E$ where $E$ is the error polynomial with
$\rank(E) = \frac{m-k}{2} = 2$. 
    \[
    \Dick_\GG(Y) =  
\underbrace{{\scriptsize{  \begin{pmatrix}
      c_0 & 0 & 0 & 0 & 0 & c_2^{q^5} & c_1^{q^6} \\
      c_1 & c_0^q & 0 & 0 & 0 & 0 & c_2^{q^6} \\
      c_2 & c_1^q & c_0^{q^2} & 0 & 0 & 0 & 0 \\
      0 & c_2^q & c_1^{q^2} & c_0^{q^3} & 0 & 0 & 0 \\
      0 & 0 & c_2^{q^2} & c_1^{q^3} & c_0^{q^4} & 0 & 0 \\
      0 & 0 & 0 & c_2^{q^3} & c_1^{q^4} & c_0^{q^5} & 0 \\
      0 & 0 & 0 & 0 & c_2^{q^4} & c_1^{q^5} & c_0^{q^6} \\
    \end{pmatrix}
  }}}_{\Dick_\GG (C)}
+
\underbrace{{\scriptsize{  \begin{pmatrix}
      e_0 & e_6^q & e_5^{q^2} & e_4^{q^{3}} & e_3^{q^4} & e_2^{q^5} & e_1^{q^6} \\
      e_1 & e_0^q & e_6^{q^2} & e_5^{q^{3}} & e_4^{q^4} & e_3^{q^5} & e_2^{q^6} \\
      e_2 & e_1^q & e_0^{q^2} & e_6^{q^3} & e_5^{q^4} & e_4^{q^5} & e_3^{q^6} \\
      e_3 & e_2^q & e_1^{q^2} & e_0^{q^3} & e_6^{q^4} & e_5^{q^5} & e_4^{q^6} \\
      e_4 & e_3^q & e_2^{q^2} & e_1^{q^3} & e_0^{q^4} & e_6^{q^5} & e_5^{q^6} \\
      e_5 & e_4^q & e_3^{q^2} & e_2^{q^3} & e_1^{q^4} & e_0^{q^5} & e_6^{q^6} \\
      e_6 & e_5^q & e_4^{q^2} & e_3^{q^3} & e_2^{q^4} & e_1^{q^5} & e_0^{q^6} \\
    \end{pmatrix}
}}}_{\Dick_\GG(E)}.
\]
Therefore, the coefficients $e_i$ for $3 \le i \le 6$ of the error polynomial $E$ are known, as illustrated below where known entries are represented in light blue.
\[
\underbrace{{\scriptsize{\begin{pmatrix}
     e_0 & \textcolor{known}{e_6^q} & \textcolor{known}{e_5^{q^2}} & \textcolor{known}{e_4^{q^{3}}} & \textcolor{known}{e_3^{q^4}} & e_2^{q^5} & e_1^{q^6} \\
e_1 & e_0^q & \textcolor{known}{e_6^{q^2}} & \textcolor{known}{e_5^{q^{3}}} & \textcolor{known}{e_4^{q^4}} & \textcolor{known}{e_3^{q^5}} & e_2^{q^6} \\
e_2 & e_1^q & e_0^{q^2} & \textcolor{known}{e_6^{q^3}} & \textcolor{known}{e_5^{q^4}} & \textcolor{known}{e_4^{q^5}} & \textcolor{known}{e_3^{q^6}} \\
\textcolor{known}{e_3} & e_2^q & e_1^{q^2} & e_0^{q^3} & \textcolor{known}{e_6^{q^4}} & \textcolor{known}{e_5^{q^5}} & \textcolor{known}{e_4^{q^6}} \\
\textcolor{known}{e_4} & \textcolor{known}{e_3^q} & e_2^{q^2} & e_1^{q^3} & e_0^{q^4} & \textcolor{known}{e_6^{q^5}} & \textcolor{known}{e_5^{q^6}} \\
\textcolor{known}{e_5} & \textcolor{known}{e_4^q} & \textcolor{known}{e_3^{q^2}} & e_2^{q^3} & e_1^{q^4} & e_0^{q^5} & \textcolor{known}{e_6^{q^6}} \\
\textcolor{known}{e_6} & \textcolor{known}{e_5^q} & \textcolor{known}{e_4^{q^2}} & \textcolor{known}{e_3^{q^3}} & e_2^{q^4} & e_1^{q^5} & e_0^{q^6} \\
    \end{pmatrix}
}}}_{\Dick(E)}.
\]
Following the framework described in
Section~\ref{sec:dickson_framework}, we first consider a $3 \times 3$
submatrix containing a $q$--th power of $e_2$ (namely $e_2^{q^2}$) on
its top-right corner as the only unknown entry as shown in the
leftmost matrix of Figure~\ref{fig:minors_Gabidulin}. Then
$e_{2}^{q^2}$ can be recovered by solving a simple equation of degree
$1$ which permits to deduce $e_2$. The next unknown coefficient $e_1$
lies on the diagonal above the diagonal of $e_2$. Thus it is possible
to find a $3 \times 3$ submatrix containing $e_1$ (by shifting the
previously considered matrix by one row). It similarly recovers $e_1$
and we repeat the process for $e_0$ and we recover
$E$. Figure~\ref{fig:minors_Gabidulin} describes the sequence of
involved $3 \times 3$ minors.
    
    \begin{figure}[!h]
      \centering
    \begin{tikzpicture}[>=latex]
\matrix (A) [matrix of math nodes,nodes = {font=\scriptsize},
             row sep=0.6pt, column sep=0.6pt, inner sep=0.6pt,left delimiter  = (,right delimiter = )] at (0,0)
{\textcolor{unknown}{e_0} & \textcolor{known}{e_6^q} & \textcolor{known}{e_5^{q^2}} & \textcolor{known}{e_4^{q^{3}}} & \textcolor{known}{e_3^{q^4}} & \textcolor{unknown}{e_2^{q^5}} & \textcolor{unknown}{e_1^{q^6}} \\
  \textcolor{unknown}{e_1} & \textcolor{unknown}{e_0^q} & \textcolor{known}{e_6^{q^2}} & \textcolor{known}{e_5^{q^{3}}} & \textcolor{known}{e_4^{q^4}} & \textcolor{known}{e_3^{q^5}} & \textcolor{unknown}{e_2^{q^6}} \\
  \textcolor{unknown}{e_2} & \textcolor{unknown}{e_1^q} & \textcolor{unknown}{e_0^{q^2}} & \textcolor{known}{e_6^{q^3}} & \textcolor{known}{e_5^{q^4}} & \textcolor{known}{e_4^{q^5}} & \textcolor{known}{e_3^{q^6}} \\
  \textcolor{known}{e_3} & \textcolor{unknown}{e_2^q} & \textcolor{unknown}{e_1^{q^2}} & \textcolor{unknown}{e_0^{q^3}} & \textcolor{known}{e_6^{q^4}} & \textcolor{known}{e_5^{q^5}} & \textcolor{known}{e_4^{q^6}} \\
  \textcolor{known}{e_4} & \textcolor{known}{e_3^q} & \textcolor{unknown}{e_2^{q^2}} & \textcolor{unknown}{e_1^{q^3}} & \textcolor{unknown}{e_0^{q^4}} & \textcolor{known}{e_6^{q^5}} & \textcolor{known}{e_5^{q^6}} \\
  \textcolor{known}{e_5} & \textcolor{known}{e_4^q} & \textcolor{known}{e_3^{q^2}} & \textcolor{unknown}{e_2^{q^3}} & \textcolor{unknown}{e_1^{q^4}} & \textcolor{unknown}{e_0^{q^5}} & \textcolor{known}{e_6^{q^6}} \\
  \textcolor{known}{e_6} & \textcolor{known}{e_5^q} & \textcolor{known}{e_4^{q^2}} & \textcolor{known}{e_3^{q^3}} & \textcolor{unknown}{e_2^{q^4}} & \textcolor{unknown}{e_1^{q^5}} & \textcolor{unknown}{e_0^{q^6}} \\
};

\matrix (B) [matrix of math nodes,nodes = {font=\scriptsize},
             row sep=0.6pt, column sep=0.6pt, inner sep=0.6pt,left delimiter  = (,right delimiter =)] at (4,0)
{
 \textcolor{unknown}{e_0} & \textcolor{known}{e_6^q} & \textcolor{known}{e_5^{q^2}} & \textcolor{known}{e_4^{q^{3}}} & \textcolor{known}{e_3^{q^4}} & \textcolor{known}{e_2^{q^5}} & \textcolor{unknown}{e_1^{q^6}} \\
 \textcolor{unknown}{e_1} & \textcolor{unknown}{e_0^q} & \textcolor{known}{e_6^{q^2}} & \textcolor{known}{e_5^{q^{3}}} & \textcolor{known}{e_4^{q^4}} & \textcolor{known}{e_3^{q^5}} & \textcolor{known}{e_2^{q^6}} \\
 \textcolor{known}{e_2} & \textcolor{unknown}{e_1^q} & \textcolor{unknown}{e_0^{q^2}} & \textcolor{known}{e_6^{q^3}} & \textcolor{known}{e_5^{q^4}} & \textcolor{known}{e_4^{q^5}} & \textcolor{known}{e_3^{q^6}} \\
 \textcolor{known}{e_3} & \textcolor{known}{e_2^q} & \textcolor{unknown}{e_1^{q^2}} & \textcolor{unknown}{e_0^{q^3}} & \textcolor{known}{e_6^{q^4}} & \textcolor{known}{e_5^{q^5}} & \textcolor{known}{e_4^{q^6}} \\
 \textcolor{known}{e_4} & \textcolor{known}{e_3^q} & \textcolor{known}{e_2^{q^2}} & \textcolor{unknown}{e_1^{q^3}} & \textcolor{unknown}{e_0^{q^4}} & \textcolor{known}{e_6^{q^5}} & \textcolor{known}{e_5^{q^6}} \\
 \textcolor{known}{e_5} & \textcolor{known}{e_4^q} & \textcolor{known}{e_3^{q^2}} & \textcolor{known}{e_2^{q^3}} & \textcolor{unknown}{e_1^{q^4}} & \textcolor{unknown}{e_0^{q^5}} & \textcolor{known}{e_6^{q^6}} \\
 \textcolor{known}{e_6} & \textcolor{known}{e_5^q} & \textcolor{known}{e_4^{q^2}} & \textcolor{known}{e_3^{q^3}} & \textcolor{known}{e_2^{q^4}} & \textcolor{unknown}{e_1^{q^5}} & \textcolor{unknown}{e_0^{q^6}} \\
};

\matrix (C) [matrix of math nodes,nodes = {font=\scriptsize},
             row sep=0.6pt, column sep=0.6pt, inner sep=0.6pt,left delimiter  = (,right delimiter =)] at (8,0)
 {
      \textcolor{unknown}{e_0} & \textcolor{known}{e_6^q} & \textcolor{known}{e_5^{q^2}} & \textcolor{known}{e_4^{q^{3}}} & \textcolor{known}{e_3^{q^4}} & \textcolor{known}{e_2^{q^5}} & \textcolor{known}{e_1^{q^6}} \\
      \textcolor{known}{e_1} & \textcolor{unknown}{e_0^q} & \textcolor{known}{e_6^{q^2}} & \textcolor{known}{e_5^{q^{3}}} & \textcolor{known}{e_4^{q^4}} & \textcolor{known}{e_3^{q^5}} & \textcolor{known}{e_2^{q^6}} \\
      \textcolor{known}{e_2} & \textcolor{known}{e_1^q} & \textcolor{unknown}{e_0^{q^2}} & \textcolor{known}{e_6^{q^3}} & \textcolor{known}{e_5^{q^4}} & \textcolor{known}{e_4^{q^5}} & \textcolor{known}{e_3^{q^6}} \\
      \textcolor{known}{e_3} & \textcolor{known}{e_2^q} & \textcolor{known}{e_1^{q^2}} & \textcolor{unknown}{e_0^{q^3}} & \textcolor{known}{e_6^{q^4}} & \textcolor{known}{e_5^{q^5}} & \textcolor{known}{e_4^{q^6}} \\
      \textcolor{known}{e_4} & \textcolor{known}{e_3^q} & \textcolor{known}{e_2^{q^2}} & \textcolor{known}{e_1^{q^3}} & \textcolor{unknown}{e_0^{q^4}} & \textcolor{known}{e_6^{q^5}} & \textcolor{known}{e_5^{q^6}} \\
      \textcolor{known}{e_5} & \textcolor{known}{e_4^q} & \textcolor{known}{e_3^{q^2}} & \textcolor{known}{e_2^{q^3}} & \textcolor{known}{e_1^{q^4}} & \textcolor{unknown}{e_0^{q^5}} & \textcolor{known}{e_6^{q^6}} \\
      \textcolor{known}{e_6} & \textcolor{known}{e_5^q} & \textcolor{known}{e_4^{q^2}} & \textcolor{known}{e_3^{q^3}} & \textcolor{known}{e_2^{q^4}} & \textcolor{known}{e_1^{q^5}} & \textcolor{unknown}{e_0^{q^6}} \\   
};

\matrix (D) [matrix of math nodes,nodes = {font=\scriptsize},
             row sep=0.6pt, column sep=0.6pt, inner sep=0.6pt,left delimiter  = (,right delimiter =)] at (12,0)
       {
       \textcolor{known}{e_0} & \textcolor{known}{e_6^q} & \textcolor{known}{e_5^{q^2}} & \textcolor{known}{e_4^{q^{3}}} & \textcolor{known}{e_3^{q^4}} & \textcolor{known}{e_2^{q^5}} & \textcolor{known}{e_1^{q^6}} \\
       \textcolor{known}{e_1} & \textcolor{known}{e_0^q} & \textcolor{known}{e_6^{q^2}} & \textcolor{known}{e_5^{q^{3}}} & \textcolor{known}{e_4^{q^4}} & \textcolor{known}{e_3^{q^5}} & \textcolor{known}{e_2^{q^6}} \\
       \textcolor{known}{e_2} & \textcolor{known}{e_1^q} & \textcolor{known}{e_0^{q^2}} & \textcolor{known}{e_6^{q^3}} & \textcolor{known}{e_5^{q^4}} & \textcolor{known}{e_4^{q^5}} & \textcolor{known}{e_3^{q^6}} \\
       \textcolor{known}{e_3} & \textcolor{known}{e_2^q} & \textcolor{known}{e_1^{q^2}} & \textcolor{known}{e_0^{q^3}} & \textcolor{known}{e_6^{q^4}} & \textcolor{known}{e_5^{q^5}} & \textcolor{known}{e_4^{q^6}} \\
       \textcolor{known}{e_4} & \textcolor{known}{e_3^q} & \textcolor{known}{e_2^{q^2}} & \textcolor{known}{e_1^{q^3}} & \textcolor{known}{e_0^{q^4}} & \textcolor{known}{e_6^{q^5}} & \textcolor{known}{e_5^{q^6}} \\
       \textcolor{known}{e_5} & \textcolor{known}{e_4^q} & \textcolor{known}{e_3^{q^2}} & \textcolor{known}{e_2^{q^3}} & \textcolor{known}{e_1^{q^4}} & \textcolor{known}{e_0^{q^5}} & \textcolor{known}{e_6^{q^6}} \\
       \textcolor{known}{e_6} & \textcolor{known}{e_5^q} & \textcolor{known}{e_4^{q^2}} & \textcolor{known}{e_3^{q^3}} & \textcolor{known}{e_2^{q^4}} & \textcolor{known}{e_1^{q^5}} & \textcolor{known}{e_0^{q^6}} \\
      };
\draw[->] (1.8,0) -- (2.2,0);
\draw[->] (5.8,0) -- (6.2,0);
\draw[->] (9.8,0) -- (10.2,0);
\draw (-1.6,-1.6) rectangle (-.37,-.22);
\draw (2.4,-1.1) rectangle (3.6,0.22);
\draw (6.4,-.66) rectangle (7.6,0.67);
\end{tikzpicture}
\caption{The sequence of minors that permit to recover the unknown coefficients. At each step, unknown coefficients are in black font.}
\label{fig:minors_Gabidulin}
\end{figure}

\subsubsection{The general case}\label{Gcyclic}
Consider now an arbitrary Gabidulin code of length $m$ of dimension $k$. Let
$E= \sum_{i=0}^{m-1} e_i x^{q^i}$ with $e_i \in \LL$ be the error
polynomial with $\rank(E) = t = \lfloor \frac{d-1}{2} \rfloor$. The coefficients
$e_i$'s are known for $i = k, \ldots, m-1$, which appears in the
unshaded part as illustrated on the left-side of
Figure~\ref{Fig:decoding_Gabidulin}. Note that the largest square
matrix that can be drawn in that unshaded part has order
$t= \lfloor\frac{d-1}{2}\rfloor$.
\begin{figure}[h]
\centering
    \begin{tikzpicture}[scale=0.8]
    \draw[draw=black] (0,0) rectangle ++(5,5); 
    \draw[draw=black,fill=gray,opacity=0.3] (0,5) -- (0,3) -- (3,0) --(5,0) -- cycle;
    \draw[draw=black,fill=gray,opacity=0.3] (3.5,5) -- (5,5) -- (5,3.5) -- cycle;
    \draw[draw=black,fill=gray,opacity=0.05] (0,3) -- (0,0) -- (3,0)-- cycle;
    \draw[draw=black,fill=gray,opacity=0.05] (0,5) -- (3.5,5) -- (5,3.5)-- (5,0)--cycle;
     \draw[draw=black, thick, dotted] (0,0) rectangle ++(1.63,1.63);
\draw[dotted, black] (0,0) rectangle ++(1.5,1.5);
\draw[<->] (-0.3,0) -- (-0.3,1.5) node[midway, left] {$\lfloor \frac{d-1}{2}\rfloor $};

     \draw[<->] (-0.15,0) -- (-0.15,3) node at (-0.35,2.4) {$d$};
\draw[->, thick, rounded corners=10pt] 
        (4.3,0.5) to [bend right = 40] (6,2) node[above, align=center] {Unknown\\coefficients};
    \draw[->, thick, rounded corners=10pt] 
        (3,4) to [bend left = -80] (-1, 4.7) node[below, align=center] {Known \\coefficients};  
    \draw[->, thick, rounded corners=10pt] 
        (1,1.3) to [bend left = 50] (-1, 3.7) node[above, align=center] {};   
    \draw[->, thick, rounded corners=10pt] 
        (4.8,4.8) to [bend right = -30] (6,3) node[above, align=center] {}; 
    \draw (0,0) node at (0.25,4.8) {$e_0$} node at (5.1,.3) {$e_0^{q^{m-1}}$};
\draw (0,0) node at (1.85,1.6) {$e_{k-1}^{q^t}$}; 
\draw (0,0) node at (0.4,3.1) {$e_{k-1}$} node at (3.3,.3) {$e_{k-1}^{q^{m-k}}$};
     \end{tikzpicture}
     \quad
    \begin{tikzpicture}[scale=0.8]
\draw[draw=black] (0,0) rectangle ++(5,5); 
    \draw[draw=black,fill=gray,opacity=0.3] (0,5) -- (0,3) -- (3,0) --(5,0) -- cycle;
    \draw[draw=black,fill=gray,opacity=0.3] (3.5,5) -- (5,5) -- (5,3.5) -- cycle;
    \draw[draw=black,fill=gray,opacity=0.05] (0,3) -- (0,0) -- (3,0)-- cycle;
    \draw[draw=black,fill=gray,opacity=0.05] ((0,5) -- (3.5,5) -- (5,3.5)-- (5,0)--cycle;
     \draw[draw=black, thick, dotted] (0,0) rectangle ++(1.63,1.63);
     \draw[dotted, black] (0,0) rectangle ++(1.5,1.5);
     \draw[->] (1.83,1.83) -- (2.1,2.1);
     \draw[<->] (-0.3,0) -- (-0.3,1.5) node[midway, left] {$\lfloor \frac{d-1}{2}\rfloor $};
\draw[draw=black, thick, rounded corners=5pt] (0,3) -- (0,3.3) -- (3.3,0) --(3,0) -- cycle;
     \draw[->, thick, rounded corners=10pt] 
        (0.2,3) to [bend left = -80] (-0.5, 3) node[below, align=center] {$e_{k-1}^{q^*}$}; 
    \draw[draw=black, thick, rounded corners=5pt] (0,4.27) -- (0,4.57) -- (4.57,0) --(4.27,0) -- cycle;
    \draw[->, thick, rounded corners=10pt] 
        (0.2,4.27) to [bend left = -80] (-0.5, 4) node[below, align=center] {$e_{1}^{q^*}$}; 
    \draw[draw=black, thick, rounded corners=5pt] (0,4.6) -- (0,4.9) -- (4.9,0) --(4.6,0) -- cycle;
    \draw[->, thick, rounded corners=10pt] 
        (0.2,4.6) to [bend left = -80] (-0.5, 4.75) node[anchor = north, align=center] at (-0.5, 4.9) {$e_{0}^{q^*}$};
      \end{tikzpicture}
      \caption{Description of the algorithm in the general case}
      \label{Fig:decoding_Gabidulin}
\end{figure}

Now as explained in Section~\ref{sec:Gabi_1st_example}, we can find a
square submatrix of $\Dick_\GG(E)$ order $t+1$ that contains
$e_{k-1}^{q^t}$ at the top-right corner and such that all the other
entries are  known. More precisely, to recover $e_{k-1}$, we
consider the submatrix $\mathbf{D}_{I,J}$ where
$I = [k+t, k+2t]$ and $J= [1, t+1]$. Since
$\Dick_\GG(E)$ has rank $t$, then $\det(\mathbf{D}_{I,J})=0$ and this
matrix is a degree $1$ polynomial in $e_{k-1}^{q^t}$ whose leading
coefficient is a $t\times t$ minor of $\Dick_\GG(E)$ which, from
Corollary \ref{consecutive} is nonzero. Thus, $e_{k-1}^{q^t}$ can be
recovered by solving an affine equation which yields $e_{k-1}$.  Then,
iteratively shifting the submatrix $D_{I,J}$ by one column to the
right or by one row to the top and applying the same principle, we
recover the other unknown coefficients of $E$.

\subsubsection{Complexity}
Let us denote by $\mathcal{M}(\Fqm/\Fq)$ the best complexity upper bound in terms of
operations in $\Fq$ that costs a multiplication or a division in
$\Fqm$. Note that, due to \cite[Cor.~11.11]{GG13} one can take
$\mathcal{M}(\Fqm/\Fq) = \OO(m \log m) = \widetilde{\OO}(m)$.  Also, we denote by
$\omega$ the complexity exponent of linear algebra operations (
$n \times n$ matrix multiplications, matrix inversion, Gaussian
elimination, \emph{etc.}).

\begin{thm}
  The algorithm described in Section~\ref{Gcyclic} corrects up to $t = \lfloor \frac{n-k}2 \rfloor$
  errors on a Gabidulin code of dimension $k$ over $\Fqm$ in a time complexity of
  \[\OO\Big(k\mathcal{M}(\Fqm/\Fq)(t^\omega + m\log q)\Big) \quad \text{operations in }\Fq.\]
  Moreover,
  if $\log q = \Omega(t^{\omega})$, then the complexity can be turned to
  $\OO(km^2t^{\omega})$ operations in $\Fq$.
\end{thm}

\begin{proof}
  The algorithm's dominant costs consist in the computation of $k$
  consecutive determinants of $t \times t$ matrices with entries in
  $\Fqm$ and $km$ evaluations of the Frobenius map $x \mapsto x^q$
  ($m$ applications per unknown coefficient of $E$).

  The cost of computing the determinants is $\OO(k t^{\omega})$  operations in $\Fqm$
  and hence $\OO(k\mathcal{M}(\Fqm/\Fq)t^{\omega})$ operations in $\Fq$.  For the
  calculation of the Frobenius one can proceed in two different
  manners. Either we raise to the power $q$, which using fast
  exponentiation costs $\OO(\log q)$ operations in $\Fqm$ and hence
  $\OO(\mathcal{M} (\Fqm/\Fq)\log q)$ operations in $\Fq$. This leads to a complexity
  in $\OO(k\mathcal{M}(\Fqm/\Fq)(t^\omega + m\log q))$. Or, we can represent elements
  of $\Fqm$ in a normal basis over $\Fq$.  In this situation the
  Frobenius becomes a single shift on the entries and hence costs
  $\OO(m)$ operations in $\Fq$. However, when choosing such a normal
  basis, one cannot expect to use fast multiplication and should take
  $\OO(m^2)$ operations in $\Fq$ for the cost of multiplications in
  $\Fqm$.  This leads to an overall complexity in $\OO(km^2t^\omega)$
  since the cost of Gaussian eliminations $\OO(km^2 t^{\omega})$ will dominate
  the $\OO(km^2)$ to apply $km$ times the Frobenius.
  The former overall complexity turns out to be better than the
  $\OO(k\mathcal{M}(\Fqm/\Fq)(t^\omega + m\log q))$ whenever
  $\log(q) = \Omega (t^{\omega})$.
\end{proof}

\begin{rem}
The decoding algorithm based on minor cancellations of Dickson matrices illustrated above works for Gabidulin codes over arbitrary cyclic Galois extensions \cite{ALR18} exactly the same way. If $\LL/\K$ is a cyclic Galois extension of degree $m$, then the complexity of Dickson matrix-based decoding of Gabidulin codes over $\LL/\K$ of length $m$ and dimension $k$ is $\OO(k \mathcal{M}(\LL/\K)t^{\omega} + k m^3)$ operations $\K$. Indeed, the cost of computing $k$ many determinants of $t \times t$ matrices is $\OO(k \mathcal{M}(\LL/\K)t^{\omega})$ operations in $\K$ and the applying an element in the Galois group can be performed in $\OO(m^2)$ operations in $\K$. 
\end{rem}

\subsection{Decoding for Reed-Solomon codes}
We conclude this section with a side remark: this approach based on
minor cancellation can actually be used even to decode Reed--Solomon
codes. Here we restrict to cyclic Reed--Solomon codes even if the
approach may be extended to the general case.
Consider $\alpha \in \Fq^\times$ that generates the multiplicative
group of $\Fq$. Set $n \eqdef q-1$ and define
\[
    \mathbf{RS}_k \eqdef \{(f(1), f(\alpha), f(\alpha^2), \dots, f(\alpha^{q-2})) ~:~ f \in \Fq[X],\ \deg (f) < k\} \subseteq \Fq^n.
  \]
 Denote by $\wt(\cdot)$ the Hamming weight and suppose we are given,
  \begin{equation}\label{eq:decoding_RS}
    \yv = \cv + \ev,\quad \text{where}\quad \cv \in \mathbf{RS}_k\quad
    \text{and}\quad \wt(\ev) = \frac{n-k}{2}\cdot
  \end{equation}
  The Chinese
  Remainder Theorem induces an isomorphism
\begin{equation}\label{eq:ev}
   \textrm{ev} :  \map{\Fq[X]/(X^n-1)}{\Fq^n}{f}{(f(1),f(\alpha), \dots, f(\alpha^{q-2})).}
 \end{equation}
 When dealing with the decoding of Reed--Solomon codes,
 elements are represented as vectors in $\Fq^n$ and the isomorphism
 \eqref{eq:ev} above is explicit in the two directions : multiple
 evaluation in the direct sense and Lagrange interpolation in
 the converse direction.
Therefore, the decoding problem for Reed--Solomon codes can be
reformulated in a constructive manner as follows. Given $y(X) \in \Fq[X]/(X^n-1)$, find
$c,e \in \Fq[X]$ satisfying
\begin{equation}\label{eq:decoding_RS_poly}
  y(X) \equiv c(X) + e(X) \mod (X^n-1), \quad \text{such\ that}\quad
  \deg c < k,\quad \text{and}\quad \wt (\text{ev}(e)) = \frac{n-k}{2}\cdot
\end{equation}

Until the end of this section, we denote by $y(X), c(X), e(X)$ the
elements of $\Fq[X]/(X^n-1)$ that evaluate respectively to $\yv, \cv, \ev$
via the isomorphism $\textrm{ev}$ of \eqref{eq:ev}.

The similarity with rank metric codes lies in the fact that elements
of $\Fq[X]/(X^n-1)$ can be associated to a circulant matrix.  More
precisely, there is a ring isomorphism between  $\Fq[X]/{(X^n -1 ) }$ and the
ring of $n\times n$ circulant matrices over $\Fq$ given by
\begin{equation}\label{RS}
\begin{aligned}
  \mathbf{Mat} : \sum_{i =0}^{n-1} c_i X^i \longmapsto
  \begin{pmatrix}
    c_0 & c_{n-1} &\ldots &c_1\\
    c_1 & c_0 &\ldots &c_2\\
    \vdots &\ddots &\ddots &\vdots\\
    c_{n-1} &c_{n-2}&\ldots &c_0
\end{pmatrix}.
\end{aligned}
\end{equation}
Similarly to Dickson matrices, $\mathbf{Mat} (c)$ represents the
multiplication by $c$ map in the monomial basis of $\Fq[X]/(X^n
-1)$. Moreover, this isomorphism has the following metric property.

\begin{pro}\label{Circulant}
  Let $n = q-1$ and $P = \sum_{i=0}^{n-1} p_i X^i  \in \Fq[X]$.
  Then,
  \[
    \wt(\textrm{ev} (P)) = \rk (\mathbf{Mat}(P)).
  \]
  Moreover, any consecutive $\rk (\mathbf{Mat}(P))$ columns of
  $\mathbf{Mat}(P)$ are $\Fq$-linearly independent.
\end{pro}

\begin{proof} Consider the matrix
  \[\setlength\arraycolsep{3pt}
    \mathbf{A} \eqdef \begin{pmatrix}
        1 &1 &\ldots &1\\
        1 & \alpha & \ldots & \alpha^{n-1}\\
        \vdots &\vdots  &\ddots &\vdots\\
        1 &\alpha^{n-1} &\ldots &\alpha^{(n-1)^2}
    \end{pmatrix}
    \underbrace{\begin{pmatrix}
         p_0 & p_{n-1} &\ldots &p_1\\
    p_1 & p_0 &\ldots &p_2\\
    \vdots &\ddots &\ddots &\vdots\\
    p_{n-1} &p_{n-2}&\ldots &p_0
  \end{pmatrix}}_{\mathbf{Mat}(P)} = \begin{pmatrix}
    P(1) & P(1) & \ldots & P(1)\\
    P(\alpha) & \alpha P(\alpha) & \ldots & \alpha^{n-1}P(\alpha)\\
    \vdots &\vdots &\ddots &\vdots\\
    P(\alpha^{n-1}) & \alpha^{n-1} P(\alpha^{n-1}) &\ldots &\alpha^{(n-1)^2}
    P(\alpha^{n-1})
    \end{pmatrix}.
  \]
  First, since the left-hand term of the product defining
  $\mathbf{A}$ is a nonsingular Vandermonde matrix, then
  $\rk(\mathbf{A}) = \rk (\mathbf{Mat}(P))$. Second, by the very
  definition of $\wt(\text{ev}(P))$ we obtain that exactly
  $\wt(\text{ev}(P))$, rows of $\mathbf{A}$ are nonzero.
  Set $s \eqdef \wt(\text{ev}(P))$ and
  $I = \{i_0, \dots , i_{s-1}\} \subseteq [0, \dots, n-1]$ be the
  indexes of nonzero rows of $\mathbf{A}$ and
  $J = [a, a+s-1] \subseteq [0, n-1]$ be a subset of
  consecutive elements. Then,
  \[\det \mathbf{A}_{I,J} =  P(\alpha^{i_0})\cdots P(\alpha^{i_{s-1}}) \alpha^{a(i_0+\cdots + i_{s-1})} \begin{vmatrix}
1 & \alpha^{i_0} & \alpha^{2i_0} &  \cdots &\alpha^{(s-1)i_0}\\
\vdots & \vdots & \vdots & & \vdots \\
1 & \alpha^{i_{s-1}} & \alpha^{2i_{s-1}} & \cdots &\alpha^{(s-1)i_{s-1}}
\end{vmatrix} \neq 0.\] In summary, $\mathbf{A}$ has exactly $s$
nonzero rows and a nonzero $s \times s$ minor. Thus, $\mathbf{A}$ has
rank $s = \wt (\text{ev}(P))$ and so has $\mathbf{Mat}(P)$.  Moreover,
the fact that $\det(\mathbf{A}_{I,J})$ does not vanish for any set $J$
of consecutive columns entails that the corresponding columns of
$\mathbf{Mat}(P)$ are linearly independent.
\end{proof}

Therefore, one can decode Reed--Solomon codes using the previous
statement in the very same way as for Gabidulin codes. The decoding
problem (\ref{eq:decoding_RS}) reformulated in terms of polynomials
(\ref{eq:decoding_RS_poly}) can be expressed in terms of a rank
problem:
\[
  \mathbf{Mat}(y) = \mathbf{Mat}(c) + \mathbf{Mat}(e), \quad
  \text{where}\quad \deg(c) < k,\quad \text{and}\quad
  \rk (\mathbf{Mat}(e)) = t = \frac{n-k}{2}\cdot
\]
Since
$\deg (c) < k$, we deduce the top coefficients $e_k, \dots, e_{n-1}$
of $e$ which are nothing but those of $y$. Next, we can compute the unknown
entries of $\mathbf{Mat}(e)$ by iteratively solving degree $1$
equations corresponding to $(t+1)\times (t+1)$ minors cancellations.
Details are left to the reader.

 \section{{Decoding \texorpdfstring{$\TH$}{}--Reed-Muller codes}}\label{sec:Dickson_decoding}
In this section, we present a decoding algorithm for $\TH$-Reed-Muller
codes based on $\GG$-Dickson matrices that corrects any error up to
rank equal to half the minimum distance answering an open question in
\cite{ACLN}. Throughout this section, we consider an arbitrary but
fixed $\TH$-Reed-Muller code $\RM_{\TH}(r, \nn)$ of order $r$, and
type $\nn = (n_1, \ldots, n_m)$, where $r$ and $m$ are positive
integers such that
\[
  \nn \in \NN^m,\quad \text{with}\quad n_1 \ge n_2 \ge \cdots \ge n_m \ge 2
  \quad \text{and}\quad r \le \sum_{i=1}^m (n_i -1).
\]
We follow the notations declared in Section~\ref{sec:Rank_RM} and in
particular in Definition~\ref{thpoly}. Additionally, we set
\[
  N \eqdef \prod_{i=1}^m n_i.
\]
Recall that, according to Theorem~\ref{min}, writing
$r = \ell + \sum_{i=s+1}^m (n_i-1)$ for uniquely defined integers $\ell, s$, then the code
$\RM_{\TH}(r, \nn)$ has minimum distance
\[
  d = (n_s-\ell) \prod_{i=1}^{s-1} n_i \quad
  \text{and we fix}\quad t \le \left\lfloor \frac{d-1}{2}\right\rfloor.
\]
Finally, we denote by $k$ the dimension of $\RM_{\TH}(r,m)$.

Let us recall our decoding problem in the framework of $\LL[\GG]$-codes \cite{ACLN}.

\begin{pb}
Given $Y \in \LL[\GG]$ such that $Y = C+E$ for some
$C \in \RM_{\boldsymbol{\th}}(r,\textbf{n})$ and $E \in \LL[\GG]$ with $\rank(E) = t \le \lfloor \frac{d-1}{2} \rfloor$, recover the pair $(C,E)$.
\end{pb}

Our decoding procedure consists of iterative recovery of unknown coefficients of the error $\TH$--polynomial $E$ by a majority voting method applied on the $\GG$-Dickson matrix of $E$. For this, one key component will be the shape of the $\GG$-Dickson matrix, or more precisely the positions of the unknown coefficients in the matrix, which we describe next.

\subsection{The shape of a \texorpdfstring{$\GG$}{G}-Dickson matrix}
Let
$\GG = {\mathbb{Z}}/ {n_1 \mathbb{Z}} \times \cdots \times
{\mathbb{Z}}/{n_m \mathbb{Z}}$ and $\TH = (\th_1, \ldots, \th_m)$ be a
set of generators of $\GG$.  The set $\Lambda(\nn)$ introduced in
\eqref{eq:LambdaTheta} is ordered with the reverse lexicographic
ordering as follows.  For
$\ii = (i_1, \ldots, i_m),~ \jj = (j_1, \ldots, j_m) \in
\Lambda(\nn)$,
\[
  \ii \preceq_{revlex} \jj \quad
  \text{iff\ for\ some\ } s \in [1,m],\quad
  i_s < j_s \quad \text{and}\quad \forall t > s,\ i_t = j_t.
\]
For brevity, we will omit the $revlex$ subscript from
now on and only denote it as $\preceq$. Moreover, if $\ii \preceq \jj$
and $\ii \neq \jj$, then we simply write $\ii \prec \jj$.

Since any element $\gg \in \GG$ has a unique representative
$\ii \in \Lambda(\nn)$ such that
$\gg = \TH^{\ii} = \th_1^{i_1} \cdots \th_m^{i_m}$, the reverse
lexicographic order $\preceq$ induces a total order on $\GG$ that
we also denote by $\preceq$.
Therefore, we will regularly denote the elements of $\GG$ as $\gg_i$
for $i \in [0,N-1]$ ordered with respect to $\preceq$.
This can be made explicit as
follows: we denote $\TH^{\ii}$ by $\gg_{\varphi(\ii)}$ where $\varphi$
is the following bijection.
\begin{equation}\label{phi}
\begin{aligned}
  \varphi \colon \map{\Lambda(\nn)}{[0, N -1]}{(a_1,\ldots,
    a_m)}{a_1 + a_2 n_1 + a_3 n_1 n_2 + \cdots + a_m n_1 \cdots n_{m-1}.}
    \end{aligned}
    \end{equation}

    \begin{rem}\label{rem:increasing}
      The bijections $\varphi$ of \eqref{phi} and $\pinv$ are
      both strictly increasing w.r.t the orderings $\preceq_{revlex}$
      and $\le$.
    \end{rem}
    
Throughout this section, we express a
    $\boldsymbol\theta$-polynomial
    $F = \sum_{\ii \in \Lambda(\nn)} f_{\ii} \TH^{\ii} \in \LL[\GG]$ as
\begin{equation}
  F = \sum_{t=0}^{N-1} f_{t} \gg_t, \quad\text{where}\quad
  \pinv(t)=\ii \quad\text{and}\quad \gg_t = \TH^{\ii}.
\end{equation}
\begin{exa}
  Let
  $\GG = \mathbb{Z} / {3\mathbb{Z}} \times {\mathbb{Z} / {3\mathbb{Z}}}
  = \< \theta_1, \theta_2\>$ and the elements of $\GG$ with respect to
  the reverse lexicographic ordering are as follows:
\[
    \begin{array}{lll}\gg_0 = \theta_1^0 \,\theta_2^0,  &\gg_1 = \theta_1^1 \,\theta_2^0, &  \gg_2 = \theta_1^2 \,\theta_2^0,\\
                \gg_3 = \theta_1^0\, \theta_2^1, &\gg_4= \theta_1^1\, \theta_2^1,& \gg_5 = \theta_1^2 \,\theta_2^1, \\
                \gg_6 = \theta_1^0\, \theta_2^2,  &\gg_7 =
                  \theta_1^1 \,\theta_2^2, & \gg_8 = \theta_1^2
                  \,\theta_2^2.
    \end{array}
  \]
To make the coefficients more explicit, we write a $\TH$-polynomial in
 $\LL[G]$ as
 $F = \sum_{i,j=0}^{2} f_i^j \theta_1^{i} \theta_2^{j}$ and its
 $\GG$-Dickson matrix takes the form:
\begin{table}[h]
    \centering
{
\begin{tabular}{ccc|ccc|ccc}
        \yell $ f_0^0$& \rd $\gg_1(f_2^0)$ & \orn $\gg_2(f_1^0)$  & \y $\gg_3(f_0^2) $ & \y $ \gg_4(f_2^2)$ &\y $\gg_5(f_1^2)$  & \gr $\gg_6(f_0^1)$ &\gr  $\gg_7(f_2^1) $& \gr $\gg_8(f_1^1)$ \\
         \orn $f_1^0$ &\yell  $\gg_1(f_0^0)$ & \rd $\gg_2(f_2^0)$  &\y $\gg_3(f_1^2) $& \y $\gg_4(f_0^2)$ &\y $\gg_5(f_2^2)$  & \gr $\gg_6(f_1^1)$ &\gr $\gg_7(f_0^1)$ & \gr$\gg_8(f_2^1)$ \\
        \rd $f_2^0$ & \orn $\gg_1(f_1^0)$& \yell $\gg_2(f_0^0)$  &\y $\gg_3(f_2^2)$ & \y $\gg_4(f_1^2)$ &\y $\gg_5(f_0^2)$  & \gr $\gg_6(f_2^1)$ &\gr $\gg_7(f_1^1)$ & \gr $\gg_8(f_0^1)$\\
        \hline 
        \gr $f_0^1$ & \gr $\gg_1(f_2^1)$& \gr $\gg_2(f_1^1)$  &\yell $\gg_3(f_0^0)$ & \rd $\gg_4(f_2^0)$ &\orn $\gg_5(f_1^0)$  &\y $ \gg_6(f_0^2)$ & \y $\gg_7(f_2^2)$ & \y $\gg_8(f_1^2)$\\
        \gr $f_1^1$ &\gr  $\gg_1(f_0^1)$ & \gr $\gg_2(f_2^1)$ &\orn $\gg_3(f_1^0)$ & \yell $\gg_4(f_0^0)$ &\rd $\gg_5(f_2^0)$  &\y $\gg_6(f_1^2)$ & \y $\gg_7(f_0^2)$ &\y $\gg_8(f_2^2)$ \\
        \gr $f_2^1$ &\gr  $\gg_1(f_1^1)$ & \gr $\gg_2(f_0^1)$  &\rd $\gg_3(f_2^0)$ & \orn $\gg_4(f_1^0)$ &\yell $\gg_5(f_0^0)$  &\y $\gg_6(f_2^2) $& \y$ \gg_7(f_1^2)$ & \y$\gg_8(f_0^2)$\\
         \hline 
        \y $f_0^2$ &\y $\gg_1(f_2^2)$ & \y$\gg_2(f_1^2)$  &\gr $\gg_3(f_0^1)$ & \gr $\gg_4(f_2^1)$ & \gr $\gg_5(f_1^1)$  & \yell $\gg_6(f_0^0)$ & \rd $\gg_7(f_2^0)$ & \orn $\gg_8(f_1^0)$ \\
        \y $f_1^2$ & \y $\gg_1(f_0^2)$ & \y $\gg_2(f_2^2)$  & \gr $\gg_3(f_1^1)$ & \gr $\gg_4(f_0^1)$ &\gr  $\gg_5(f_2^1)$  &  \orn $\gg_6(f_1^0)$ &  \yell$\gg_7(f_0^0)$ & \rd $\gg_8(f_2^0)$\\
        \y $f_2^2$ &\y  $\gg_1(f_1^2)$ & \y $\gg_2(f_0^2)$  & \gr $\gg_3(f_2^1)$ & \gr $\gg_4(f_1^1)$ & \gr $\gg_5(f_0^1)$  & \rd $\gg_6(f_2^0)$& \orn $\gg_7(f_1^0)$ &  \yell $\gg_8(f_0^0)$
\end{tabular}
} \captionof{figure}{$\GG$-Dickson matrix of
  $F = \sum_{i,j=0}^{2} f_i^j \th_1^{i} \th^{j}$.  Compared to the
  $\GG$-Dickson matrix for $\GG$ cyclic, which is $q$-circulant, for
  $\GG \cong \ZZ/{n\ZZ} \times \ZZ/{n\ZZ}$, it is a block circulant matrix. The colors depict that the $9 \times 9$ matrix can be seen as a $3 \times 3$ block matrix where each block is a sort of circulant matrix and the blocks appear in a circular manner (up to applications of elements of the Galois group). If we ignore the applications of the Galois group elements, then the coefficients of $F$ are exactly in a block circulant form as defined in \cite[\S 4.0]{Tr73}.
}\label{table_blockcirculent}
\end{table}
 \end{exa}

\subsection{Known coefficients of $E$}\label{subsec:known_coefficients}
Knowing $Y = C+E$ with $C \in \RM_{\TH}(r,m)$, then, by definition of the
code, for any $\gg \in \GG$ with $\deg_{\TH}(\gg)>r$, $C_{\gg} = 0$ and hence
$E_{\gg} = Y_{\gg}$. Thus, for any $\gg \in \GG$ with $\deg_{\TH}(\gg) > r$,
the coefficient $E_{\gg}$ is known.

Consequently, we have a partial knowledge of the entries of the
$\GG$-Dickson matrix $\Dm_{\GG}(E)$ of the error. Moreover, from
Proposition~\ref{prop:Dickson_rank}, we know that the rank of
$\Dm_{\GG}(E)$ is bounded from above by
$t \leq \lfloor \frac{d-1}{2}\rfloor$.  The principle of our decoding
algorithm is to recover the unknown entries of $\Dm_{\GG}(E)$ by a
majority voting process thanks to two main properties: first, it is a
$\GG$-Dickson matrix and hence many of its entries are conjugate under the
action of $\GG$; second, its rank is bounded by $t$.

\subsection{The unknown coefficients along the diagonals of $\Dick_{\GG} (E)$}
Since, from Remark~\ref{rem:increasing}, the elements of $\GG$ are in
increasing bijection with elements of $\Lambda(\nn)$, we transport the
$\TH$--degree on $\Lambda (\nn)$ by denoting
\begin{equation}\label{eq:degree}
  \text{for}\ \  \ii = (i_1,\dots, i_m) \in \Lambda(\nn),\quad
  |\ii| \eqdef i_1+\cdots + i_m.
\end{equation}

Our algorithm recovers the unknown coefficients of $E$ using a majority voting method. First, we locate the unknown coefficients on $\Dick_{\GG }(E)$ and for that we consider the following notion of
\emph{diagonals} of a matrix.

\begin{defn}\label{diagonal}
  Let $i \in [0,N-1]$. Then the \emph{$i$-th diagonal} of the $N \times N$ board is the set
  \[\Delta_i \eqdef \{(i  + s, s) \colon s\in [0,N-1] \text{ and } i+s \le N-1 \}.\]
\end{defn}

A visual representation of the five diagonals of the $5 \times 5$ board, with different colors is given below.

\begin{center}
\tiny{
\begin{tikzpicture}[scale=0.5]
\draw[thick] (0,0) grid (5,5);
\draw[fill=LightSkyBlue1,opacity=1]  (0,4) -- (0,5) -- (1,5) -- (1,4) -- cycle;
\draw[fill=LightSkyBlue1]  (1,3) -- (1,4) --  (2,4) -- (2,3) --  cycle;
\draw[fill=LightSkyBlue1]  (2,3) -- (3,3) -- (3,2) -- (2,2) -- cycle;
\draw[fill=LightSkyBlue1]  (3,2) -- (4,2) -- (4,1) -- (3,1) -- cycle;
\draw[fill=LightSkyBlue1]  (4,1) -- (5,1) -- (5,0) -- (4,0) -- cycle;

\draw[fill=DarkSeaGreen2]  (0,3) -- (0,4) -- (1,4) -- (1,3) -- cycle;
\draw[fill=DarkSeaGreen2]  (1,2) -- (1,3) --  (2,3) -- (2,2) --  cycle;
\draw[fill=DarkSeaGreen2]  (2,2) -- (3,2) -- (3,1) -- (2,1) -- cycle;
\draw[fill=DarkSeaGreen2]  (3,1) -- (4,1) -- (4,0) -- (3,0) -- cycle;

\draw[fill=Azure2]  (0,2) -- (0,3) -- (1,3) -- (1,2) -- cycle;
\draw[fill=Azure2]  (1,1) -- (1,2) --  (2,2) -- (2,1) --  cycle;
\draw[fill=Azure2]  (2,1) -- (3,1) -- (3,0) -- (2,0) -- cycle;

\draw[fill=Thistle2]  (0,1) -- (0,2) -- (1,2) -- (1,1) -- cycle;
\draw[fill=Thistle2]  (1,0) -- (1,1) --  (2,1) -- (2,0) --  cycle;

\draw[fill=LavenderBlush1]  (0,0) -- (0,1) -- (1,1) -- (1,0) -- cycle;

\end{tikzpicture}}
\end{center}

We iteratively recover the unknown coefficients of the error
polynomial $E$ in a decreasing order according to the reverse
lexicographic ordering. This means that at each iteration,
we search the unknown coefficient $e_s$ of $E$ with the largest
possible index $s \in [0, N-1]$. This coefficient will be referred
to as $\efar$ and its index is updated at each iteration.

Note that, as already mentioned in Remark \ref{rem:conjugate_unknown},
we actually recover $\gg_k(e_i)$ for some known $\gg_k \in \GG$ which
is equivalent to recovering $e_i$ since $\gg_k$ is an automorphism.

\begin{exa}
  Let $\textbf{n} = (3,3)$ and $r = 1$. Then minimum rank distance
  $d = 6$. When starting the decoding process, the furthest unknown
  coefficient is $f_0^{1}$. To highlight the positions of the unknown
  coefficients $f_0^1$, $f_1^0$ and $f_0^0$, we put them in black font
  and the known ones in color.
  
  \begin{figure}[!ht]
      \centering
\begin{tikzpicture}[>=latex]
\matrix (A) [matrix of math nodes,nodes = {font=\scriptsize},
             row sep=1pt, column sep=1pt, inner sep=1pt,left delimiter  = (,right delimiter = )] at (0,0)
{\textcolor{unknown}{f_0^0} &  \textcolor{known}{\gg_1(f_2^0)} &  \textcolor{unknown}{\gg_2(f_1^0)}  &  \textcolor{known}{\gg_3(f_0^2)}  &   \textcolor{known}{\gg_4(f_2^2)} & \textcolor{known}{\gg_5(f_1^2)}  &  \textcolor{unknown}{\gg_6(f_0^1)} &  \textcolor{known}{\gg_7(f_2^1)} &  \textcolor{known}{\gg_8(f_1^1)} \\        
     \textcolor{unknown}{f_1^0} &  \textcolor{unknown}{\gg_1(f_0^0)} &  \textcolor{known}{\gg_2(f_2^0)}  & \textcolor{known}{\gg_3(f_1^2)} &  \textcolor{known}{\gg_4(f_0^2)} & \textcolor{known}{\gg_5(f_2^2)}  &  \textcolor{known}{\gg_6(f_1^1)} & \textcolor{unknown}{\gg_7(f_0^1)} &  \textcolor{known}{\gg_8(f_2^1)}\\
     \textcolor{known}{f_2^0} &  \textcolor{unknown}{\gg_1(f_1^0)} &  \textcolor{unknown}{\gg_2(f_0^0)} & \textcolor{known}{\gg_3(f_2^2)} &  \textcolor{known}{\gg_4(f_1^2)} & \textcolor{known}{\gg_5(f_0^2)}  &  \textcolor{known}{\gg_6(f_2^1)} & \textcolor{known}{\gg_7(f_1^1)} &  \textcolor{unknown}{\gg_8(f_0^1)}\\ 
     \textcolor{unknown}{f_0^1} &  \textcolor{known}{\gg_1(f_2^1)}&   \textcolor{known}{\gg_2(f_1^1)}  & \textcolor{unknown}{\gg_3(f_0^0)} &  \textcolor{known}{\gg_4(f_2^0)} & \textcolor{unknown}{\gg_5(f_1^0)}  &  \textcolor{known}{\gg_6(f_0^2)} &  \textcolor{known}{\gg_7(f_2^2)} &  \textcolor{known}{\gg_8(f_1^2)}\\
     \textcolor{known}{f_1^1} &  \textcolor{unknown}{\gg_1(f_0^1)} &  \textcolor{known}{\gg_2(f_2^1)} & \textcolor{unknown}{\gg_3(f_1^0)} &  \textcolor{unknown}{\gg_4(f_0^0)} & \textcolor{known}{\gg_5(f_2^0)}  & \textcolor{known}{\gg_6(f_1^2)} &  \textcolor{known}{\gg_7(f_0^2)} & \textcolor{known}{\gg_8(f_2^2)} \\
     \textcolor{known}{f_2^1} &  \textcolor{known}{\gg_1(f_1^1)} &  \textcolor{unknown}{\gg_2(f_0^1)}  & \textcolor{known}{\gg_3(f_2^0)} &  \textcolor{unknown}{\gg_4(f_1^0)} & \textcolor{unknown}{\gg_5(f_0^0)}  & \textcolor{known}{\gg_6(f_2^2)} &  \textcolor{known}{\gg_7(f_1^2)} & \textcolor{known}{\gg_8(f_0^2)}\\
     \textcolor{known}{f_0^2} & \textcolor{known}{\gg_1(f_2^2)} &  \textcolor{known}{\gg_2(f_1^2)}  & \textcolor{unknown}{\gg_3(f_0^1)} &  \textcolor{known}{\gg_4(f_2^1)} &  \textcolor{known}{\gg_5(f_1^1)}  &  \textcolor{unknown}{\gg_6(f_0^0)} &  \textcolor{known}{\gg_7(f_2^0)} &  \textcolor{unknown}{\gg_8(f_1^0)} \\
     \textcolor{known}{f_1^2} &  \textcolor{known}{\gg_1(f_0^2)} &  \textcolor{known}{\gg_2(f_2^2)}  &  \textcolor{known}{\gg_3(f_1^1)} &  \textcolor{unknown}{\gg_4(f_0^1)} &  \textcolor{known}{\gg_5(f_2^1)}  &   \textcolor{unknown}{\gg_6(f_1^0)} &   \textcolor{unknown}{\gg_7(f_0^0)} &  \textcolor{known}{\gg_8(f_2^0)}\\
     \textcolor{known}{f_2^2} &  \textcolor{known}{\gg_1(f_1^2)} &  \textcolor{known}{\gg_2(f_0^2)}  &  \textcolor{known}{\gg_3(f_2^1)} &  \textcolor{known}{\gg_4(f_1^1)} &  \textcolor{unknown}{\gg_5(f_0^1)}  &  \textcolor{known}{\gg_6(f_2^0)} &  \textcolor{unknown}{\gg_7(f_1^0)} &   \textcolor{unknown}{\gg_8(f_0^0)} \\
};
\end{tikzpicture}
\caption{Illustration of positions of the unknown coefficients.}
\label{fig:minors_RM}
\end{figure}

\end{exa}

\begin{rem}
  Note that $\Delta_0$ contains $f_0^0$ and its conjugates. Similarly,
  $\Delta_1$ and $\Delta_3$ contain $f_1^0$ and $f_0^1$, respectively
  with their respective conjugates. Furthermore, we observe that the
  number of occurrences of an unknown coefficient and its conjugates
  on the respective diagonal is at least $d$, \emph{i.e.}, the minimum
  rank distance of the code. We will show in the sequel that this
  happens in general too.
\end{rem}

We first describe the unknown coefficients, their occurrences (or occurrences of their conjugates) along the diagonals of $\Dick_{\GG} (E)$ which will be used for the decoding. For this sake, we frequently allow the following notation in the sequel.

\begin{nota}\label{nota:plus_in_Delta}
  Given two elements $\ii, \jj \in \Lambda(\nn)$, we denote by $\ii + \jj$
  (resp $\ii - \jj$) the unique representative of $\TH^{\ii}\TH^\jj$
  (resp. $\TH^{\ii}\TH^{-\jj}$) in $\Lambda(\nn)$.
\end{nota}

The following lemma will be useful.

\begin{lem}\label{lem:addition_without_carries}
  Let $a,b \in [0,N-1]$ such that $a+b< N$. Then
  \begin{equation}\label{eq:sub_additive}
    \pinv(a+b) \succeq \pinv(a) + \pinv(b).
  \end{equation}
  Moreover, denoting, $\pinv(a) = (a_1, \dots, a_m)$
  and $\pinv(b) = (b_1, \dots, b_m)$, the above inequality is an equality
  if and only if for any $i\in [1,m]$, $a_i+b_i < n_i$.
\end{lem}

\begin{proof}
  Set $\pinv(a) = (a_1, \dots, a_m)$ and
  $\pinv(b) = (b_1, \dots, b_m)$. Then, by definition of $\varphi$,
  \begin{equation}\label{eq:a+b}
    a+b = (a_1+b_1) + (a_2+b_2)n_1 + \cdots + (a_m+b_m) n_1 \cdots n_{m-1},
  \end{equation}
  while
  \begin{equation}\label{eq:varpĥi(a+b)}
    \varphi(\pinv(a) + \pinv(b)) = u_1 +
    u_2n_1 + \cdots + u_m n_1 \cdots n_{m-1},
  \end{equation}
  where for any $i$, $u_i$ is the unique representative of
  $a_i+b_i \mod n_i$ in $[0, n_i-1]$. Equivalently, it is the
  remainder of $a_i+b_i$ by the Euclidean division by $n_i$. In
  particular for any $i$, $a_i+b_i \geq u_i$ and equality holds if and only if
  $a_i+b_i < n_i$. This last observation
  applied to equations (\ref{eq:a+b}) and (\ref{eq:varpĥi(a+b)}) yields
  $a+b \geq \varphi (\pinv(a) + \pinv(b))$ with equality if and only if $a_i+b_i < n_i$ for all $i \in [1,m]$. Applying
  the increasing map $\pinv$ on both sides yields the result.
\end{proof}

\begin{rem}\label{rem:carries}
  Lemma~\ref{lem:addition_without_carries} can be interpreted as
  follows. The map $\pinv$ expresses integers $a,b$ in $[0, N-1]$ in
  the ``basis'' $(1, n_1, n_1 n_2, \dots, n_1 n_2 \cdots n_{m-1})$. The
  operation $\pinv(a)+ \pinv(b)$ introduced in
  Notation~\ref{nota:plus_in_Delta} consists in the
  addition in $\GG$ which is an addition ``without
  carries'' while $a+b$ is an addition in $\mathbb{Z}$, \emph{i.e.}
  with carries.
\end{rem}

\begin{lem}\label{furthest}
  Suppose the code is $\RM_{\TH}(r,\nn) \subseteq \LL[\GG]$ with minimum rank
  distance $d$ and $r = \sum_{i=s+1}^m(n_i -1) + \ell$, for uniquely
  determined $1 \le s \le m-1$ and $0 \le \ell < n_{s}-1$. When
  beginning the decoding process all the unknown coefficients
  are among the $e_{i}$'s for $i \in [0,N-d]$ and the farthest one is
  $e_{N-d}$ where,
  \[
    N-d = \varphi(0, \ldots, 0, \ell, n_{s+1}-1, \ldots, n_m-1).
  \]
\end{lem}

\begin{proof}
  According to Section~\ref{subsec:known_coefficients}, when starting
  the decoding process the known coefficients of $E$ are the $e_s$
  such that $|\pinv(s)| > r$. Hence, the farthest unknown coefficient
  is the coefficient $e_s$ such that
  \[ \pinv(s) = \max_{\prec_{revlex}} \{\ii= (i_1,\ldots, i_m)
    ~\colon~ \ii \in \Lambda(\nn), \text{ and } |\ii| \le r \},\]
  which, since $r = \sum_{i=s+1}^m(n_i -1) + \ell$ is nothing but
  $\pinv (s) = (0, \ldots, 0, \ell, n_{s+1}-1, \ldots, n_m -1)$. Let us prove that
  $s = N-d$.  From Theorem~\ref{min},
  $d = (n_s-\ell) n_1 \cdots n_{s-1}$. Hence,
  \begin{equation}\label{eq:N-d}
    N-d = \prod\limits_{i=1}^m n_i - (n_s
    -\ell)\prod\limits_{i=1}^{s-1} n_i =
     \ell \prod_{i=1}^{s-1}n_i +
    \prod_{i=1}^s
      n_i \left(\Big(\prod_{i=s+1}^m n_i\Big)-1\right).
  \end{equation}
  Using the convention that $\prod_{i=s+1}^s n_i = 1$, one rewrites
  the rightmost factor above as an alternate sum:
  \begin{equation}
    \label{eq:alternate_sum}
    \Big(\prod_{i=s+1}^m n_i\Big)-1 = \sum_{t = s+1}^{m} \Big( \prod_{i=s+1}^{t} n_i  - \prod_{i=s+1}^{t-1} n_i \Big) = \sum_{t=s+1}^{m} (n_t-1) \prod_{i=s+1}^{t-1} n_i.
  \end{equation}
  The
  combination of~(\ref{eq:N-d}) and~(\ref{eq:alternate_sum}) yields
  \[
    N-d\  =\ \ell \prod_{i=1}^{s-1}n_i + \sum_{t=s+1}^{m}  (n_t-1) \prod_{i=1}^{t-1}n_t
       \  =\ \varphi(0, \ldots, 0, \ell, n_{s+1}-1, \ldots, n_m -1).
     \]
     Finally, since $\varphi$ is an increasing map, the other
     unknown coefficients $e_i$ satisfy $i \in [0, N-d]$.
\end{proof}

Now we observe the positions of these unknown coefficients in $\Dick_{\GG}(E)$ where the elements of $\GG$ are ordered reverse lexicographically. Here is a useful description of the $(i,j)$-th entry of $\Dick_{\GG}(E)$.
\begin{lem}\label{simple}
  Let $E = \sum_{t=0}^{N-1} e_t \gg_t \in \LL [\GG]$.  For
  $i, j \in [0, N-1]$, the $(i,j)$-th entry $\mathbf{D}_{i,j}$
  of $\Dick_{\GG}(E)$ is $\gg_j(e_k)$ where $k \in [0,N-1]$ is the
  unique element such that
  $\pinv(i) = \pinv(j) + \pinv(k)$ with `$+$'
  being given by Notation~\ref{nota:plus_in_Delta}.
\end{lem}

\begin{proof}
  By Definition \ref{Dicksondefn}, for any $i,j$,
  $\mathbf{D}_{i,j} = \gg_j(e_{{\sigma_j}^{-1}(i)})$. Set
  $k={\sigma}_j^{-1}(i)$, then $\gg_j \gg_k = \gg_i$. Thus,
  $\TH^{\pinv(j)} \TH^{\pinv(k)} =
  \TH^{\pinv(i)}$, which means
  $\pinv(j) + \pinv(k) = \pinv(i)$.
\end{proof}

It easily follows from Lemma \ref{simple} that $\mathbf{D}_{\omega,0} = e_\omega$ for any $\omega \in [0, N-1]$ and that means $e_\omega$ lies on the diagonal $\Delta_{\omega}$.
In addition, Lemma~\ref{occurrences} to follow gives the number of occurrences of an unknown coefficient $e_{\omega}$ and its conjugates on the diagonal $\Delta_{\omega}$. First, we need to recall
some property of the code's minimum distance.

\begin{pro}[{\cite[Theorem 50]{ACLN}}] \label{minrankdistance}
Let $r$ be a positive integer and $\nn = (n_1, \ldots, n_m)\in \mathbb{N}^m$ be a vector such that $n_1 \ge n_2 \ge \cdots \ge n_m \ge 2$. If $d(r, \nn)$ is the minimum rank distance of $RM_{\TH}(r, \nn)$, then  
\[
d(r,\nn) = \min \Bigl\{\prod_{i=1}^m(n_i - u_i) ~ | ~ \mathbf{u} = (u_1, \ldots, u_m) \in \Lambda(\nn),~ |\mathbf{u}| \le r\Bigr\}.
\]  
\end{pro}

\begin{lem}\label{occurrences}
Let $e_\omega$ be an unknown coefficient of the $\TH$-polynomial $E$ for some $\omega \in [0, N-d]$, where $d$ is the minimum rank distance of $\RM_{\TH}(r,\nn)$. Then the number of conjugates of $e_\omega$ appearing on the diagonal $\Delta_\omega$ is at least $d$. 
\end{lem}
 
 \begin{proof}
   Let a conjugate of $e_{\omega}$ of the form $\gg_j(e_\omega)$ appear on the diagonal $\Delta_\omega$. It follows from Lemma~\ref{simple} that if $\pinv(\omega + j) = \pinv(j) + \pinv(\omega)$ for $j \in [0,N-1]$, then $\mathbf{D}_{\omega + j,j} = \gg_j(e_\omega)$.
   Writing $\pinv(\omega) = (\omega_1, \dots, \omega_m) \in \Lambda(\nn) $, then, from Lemma~\ref{lem:addition_without_carries}, the number of occurrences of a conjugate of $e_{\omega}$ on $\Delta_{\omega}$ is given by the number of $(j_1, \dots, j_m) \in \Lambda (\nn)$ such that
   for any $i \in [1,m]$, $j_i + \omega_i \leq n_i - 1$. The number of such $m$--tuples equals $\prod_{i=1}^m (n_i - \omega_i)$. Since we are considering only the unknown coefficients, \emph{i.e.},
   $0 \le |\varphi^{-1}(\omega)| \le r$, the lower bound follows directly from Proposition \ref{minrankdistance}.
 \end{proof}

For the iterative recovery of the unknown coefficients, the following lemma will be useful.

\begin{lem}\label{important}
    For any $\tau \in [0,N-1]$, the elements lying strictly below the diagonal $\Delta_{\tau}$ are of the form $\gg(e_{\vartheta})$ for some $\gg \in \GG$ and $\vartheta > \tau$.
\end{lem}

\begin{proof}
  Following Definition \ref{diagonal},
  \[
    \Delta_{\tau} = \{(\tau + p, p) \colon p \in [0,N-1] \text { and }
    \tau + p \le N-1 \}.
  \]
  Thus, any position strictly below the diagonal $\Delta_{\tau}$ is of
  the form $(\tau_1,p)$, where $\tau_1 > \tau + p$ for some
  $p \in [0, N-1]$ such that $\tau + p \le N-1$.
  Let $\vartheta \in [0, N-1]$ such that
  $\Dick_{\tau_1,p} = \gg_p (e_\vartheta)$.
  We aim to prove that $\vartheta > \tau$.

  If $\vartheta + p \geq N$ then, since $N>\tau + p$, we get the
  proof. Hence, we suppose now that $\vartheta + p < N$.
  From
  Lemma~\ref{lem:addition_without_carries}:
  \begin{equation}\label{eq:theta+p}
    \pinv(\vartheta + p) \succeq \pinv(\vartheta) +
    \pinv(p).
  \end{equation}
  Moreover, since $\Dick_{\tau_1, p} = \gg_p (e_{\vartheta})$,
  Lemma~\ref{simple} yields
  \begin{equation}\label{eq:tau_1}
    \pinv(\tau_1) = \pinv(\vartheta) + \pinv(p).
  \end{equation}
  Combining (\ref{eq:theta+p}), (\ref{eq:tau_1}) and the increasing
  property of $\varphi$, we get $\vartheta + p \geq \tau_1$, while, by
  assumption, $\tau_1 > \tau+p$. Hence $\vartheta > \tau$.
\end{proof}

\begin{cor}\label{cor:known_coeffts}
  When starting the decoding process, any entry of $\Dm_{\GG}(E)$
  lying below the diagonal $\Delta_{N-d}$ is known.
\end{cor}

\begin{proof}
  This is a direct consequence of Lemmas~\ref{furthest} and~\ref{important}.
\end{proof}

\subsection{Recovering the unknown coefficients by majority voting}
We iteratively recover the unknown coefficients $e_i$ by decreasing
indexes, starting, from Lemma~\ref{furthest}, with $\efar =
e_{N-d}$. At any step of the decoding process, some coefficients of
$E$ remain unknown and we denote by $\efar$ the farthest one,
\emph{i.e.} the $e_s$ with the largest possible index $s$. The
corresponding diagonal $\Delta_s$ will be referred to as $\Dfar$.
According to Corollary~\ref{cor:known_coeffts} any entry of
$\Dm_{\GG}(E)$ lying below $\Dfar$ is known.

Considering $\Dfar$, the first idea could consist in doing
what we did for Gabidulin codes in
Section~\ref{sec:dickson_framework}, which is taking a
$(t+1)\times (t+1)$ submatrix whose top right--hand corner lies on the
diagonal and hence contains the value $\gg (\efar)$ for some known
$\gg \in \GG$ and deduce $\efar$ by solving a minor cancellation
equation. Since in this submatrix all the entries are known but the
top right--hand one, this minor cancellation equation is of the form
\[M\gg(\efar) + \lambda = 0,\] where $M$ equals the bottom left--hand
$t\times t$ minor and $\lambda$ only depends on the known
entries. However, for the technique to succeed we need to have
$M \neq 0$, which, in the Gabidulin case, \emph{i.e.} when $\GG$ is
cyclic is guaranteed by Corollary~\ref{consecutive}. Unfortunately, in
the non--cyclic case, we do not have any guarantee that a given
$t \times t$ minor does not vanish. To circumvent this issue, instead
of considering one submatrix whose top right--hand corner equals some
conjugate of $\efar$, we consider several such matrices whose top
right--hand corner lies on $\Dfar$.

The idea of majority voting rests on the notion of
\emph{discrepancies} (or \emph{pivots}) that we introduce now.  To
establish the statements, we borrow some notions from \cite[Definition
8.7]{HP02}.

\subsubsection{Discrepancies of a matrix}
We let $\Dm \in \LL^{N\times N}$ be a matrix. For any
$(i,j) \in [0, N-1]^2$, we denote by $\Dm(i,j)$ the submatrix of $\Dm$
whose bottom left-hand corner is that of $\Dm$ and whose top
right-hand corner is $\Dm_{i,j}$ that is to say
\[
\Dm(i,j) \eqdef \{\Dm_{i',j'} ~\colon~ i \le i' \le N-1 \text{ and } 0 \le j' \le j\}. 
\]
The definition is illustrated in the figure below.
\medskip

\begin{center}
\begin{tikzpicture}
\draw[very thick] (0,0) rectangle (4,4);
\draw[thick] (0,0) rectangle (1.5,2);
\node[left] at (0,2) {$i$};
\node[below] at (1.5,0) {$j$};
\node at (.75,1) {$\Dm(i,j)$};
\node at (2.2,2.5) {$\Dm$};
\end{tikzpicture}
\end{center}

In the sequel, we will mainly be concerned by the rank of these
submatrices $\Dm(i,j)$ for $(i,j)\in [0,N-1]^2$ and we will fix the
convention that
\[
  \forall i \in [0,N-1],\ \rk (\Dm(i,-1))  \eqdef 0\qquad \text{and}\qquad
  \forall j \in [0,N-1],\ \rk (\Dm(N,j))  \eqdef 0.
\]
In short, ``empty matrices'' are assumed to have rank zero.

\begin{defn}
  A pair $(i,j)\in [0,N-1]^2$ is called a \emph{discrepancy}
  or a \emph{pivot} if
  \[
    \rk \Dm(i+1,j) = \rk \Dm(i+1,j-1) = \rk \Dm(i,j-1)
    \quad \text{and}\quad
    \rk \Dm (i,j) \neq \rk \Dm(i+1,j-1).
  \]
\end{defn}
The matrices involved in the definition are represented in the figure below.

\begin{center}
  \begin{tikzpicture}
\draw[very thick] (0,0) rectangle (4,4);
\draw[thick] (0,0) rectangle (2.7,3);
\draw[dashed] (0,0) rectangle (2.9,3.2);
\draw[dashed] (0,0) rectangle (2.9,3);
\draw[dashed] (0,0) rectangle (2.7,3.2);
\node[left, scale=.8] at (0,3.4) {$i$};
\node[left, scale=.8] at (0,3) {$i+1$};
\node[below, scale=.8] at (3.01,0) {$j$};
\node[below, scale=.8] at (2.5,0) {$j-1$};
\node[scale=.8] at (1.35,1.7) {$\Dm(i+1,j-1)$};    
\node[right, scale=.8] at (5,1.7) {$\Dm(i+1,j)$};    
\node[right, scale=.8] at (5,4.2) {$\Dm(i,j-1)$};    
\node[right, scale=.8] at (5,3.1) {$(i,j)$};
\draw[->] (5,4.2) -- (1.35,3.1);
\draw[->] (5,1.7) -- (2.8,1.7);
\draw[->] (5,3.1) -- (2.8,3.1);
\end{tikzpicture}
\end{center}

\begin{rem}
  Actually \emph{discrepancies} should be defined with respect to a given corner of the matrix which occurs in all the
  matrices $\Dm(i,j)$. In this paper, discrepancies are defined with respect to the bottom left--hand corner: matrices
  $\Dm(i,j)$ for $(i,j)\in [0, N-1]^2$ all include the $(N-1,0)$ entry of $\Dm$.
\end{rem}

\begin{rem}\label{dis}
  Note that the \emph{discrepancies} are the pivots of $D$ obtained
  via a reverse Gaussian elimination, \emph{i.e.}, we start from the
  bottom--most row of the matrix and only allow row operations of the
  form
  ``$\text{Row}_i \leftarrow \text{Row}_i + \lambda \text{Row}_j$''
  with $i<j$. In particular, we do \textbf{not} allow to swap
  rows. Equivalently, we only allow the left action of the subgroup of
  $\mathbf{GL}_N(\LL)$ composed by upper triangular matrices with only
  $1$'s on the diagonal. Such a Gaussian elimination applied on a
  matrix of rank $t$ will ultimately lead to a matrix with only $t$
  nonzero rows and the leftmost nonzero entries of these rows will
  yield the positions of the discrepancies.
\end{rem}

\begin{exa}
  Consider the matrix over $\mathbb{Q}$:
  \[
    \begin{pmatrix}
      0 & 1 & 3 & -1 \\
      2 & 2 & 1 & 1  \\
      2 & 2 & 0 & 1  \\
      0 & 1 & 2 & -1
    \end{pmatrix}.
  \]
  After performing Gaussian elimination from the bottom and without
  swapping rows we get the following reduced form where the discrepancy positions are written in bold symbols:
  \[
    \begin{pmatrix}
      0 & 0 & 0 & 0 \\
      0 & 0 & \mathbf{1} & 0  \\
      \mathbf{2} & 0 & -4 & 3  \\
      0 & \mathbf{1} & 2 & -1
    \end{pmatrix}.
  \]  
\end{exa}

The following statement is usual in the literature. We prove it for the sake of
convenience.

\begin{pro}\label{prop:discrepancies}
  Let $\Dm \in \LL^{N \times N}$. Then
  \begin{enumerate}
   \item\label{item:col} there is at most one discrepancy per column;
   \item\label{item:row} there is at most one discrepancy per row;
  \item\label{item:rk} the total number of discrepancies equals $\rk (\Dm)$.
  \end{enumerate}
\end{pro}

\begin{proof}
  Let $i, i', j \in [0,N-1]$ with $i < i'$ and suppose that
  both $(i,j)$ and $(i',j)$ are discrepancies.
  Then $\rk \Dm (i',j) \neq \rk \Dm (i',j-1)$ which entails
  that $\rk \Dm (i+1,j) \neq \rk \Dm (i+1, j-1)$ and contradicts the
  fact that $(i,j)$ is a discrepancy. This proves (\ref{item:row})
  and the proof of (\ref{item:col}) is very similar.
  
  To prove (\ref{item:rk}) consider the sequence of matrices
  $\Dm(N-1,0), \Dm(N-1,1), \dots, \Dm(N-1,N-2), \Dm(N-1,N-1)=\Dm$,
  \emph{i.e.} the sequence of submatrices consisting of the $i$
  leftmost columns for increasing $i$. Since the rank of a matrix is
  larger than the rank of any of its submatrices, we deduce that if a
  pair $(i,j)$ is a discrepancy, then
  $\rk \Dm(N-1,j-1) = \rk \Dm(N-1,j)+1$. Conversely, such a rank drop
  occurs only if there is a discrepancy in the $j$--th column. Since
  there is at most one discrepancy per column, we deduce that the
  number of discrepancies equals $\rk (\Dm)$.
\end{proof}

\subsubsection{The majority voting process}
From now on and until the end of the current section, for the sake of
convenience, the matrix $\Dm_{\GG}(E)$ is denoted by $\Dm$. Recall that
we suppose we know a part of this matrix and aim to recover $\efar$:
the unknown coefficient $e_s$ with the largest index $s$.
Recall that we denoted by $\Dfar$ the diagonal $\Delta_s$ and that:
  \begin{itemize}
  \item From Lemma~\ref{occurrences}, at least $d$ entries of $\Dfar$
    are conjugates of $\efar$ and the positions of these entries are
    known;
  \item From Corollary~\ref{cor:known_coeffts}, the entries of
    $\Dick_{\GG}(E)$ strictly below $\Dfar$ are known.
  \end{itemize}

\subsection*{Majority votes for the correct $\efar$.}

\begin{defn}[Candidates]\label{def:candidate}
  A position $(i, j)$ on the diagonal $\Dfar$ is called a
  \emph{candidate} if it satisfies the two following conditions:
  \begin{enumerate}[(i)]
  \item $\Dm_{i,j}$ is a conjugate of $\efar$;
  \item $\Dm(i+1,j)$, $\Dm(i,j-1)$ and $\Dm(i+1,j-1)$ have the same
    rank.
  \end{enumerate}
  Otherwise, $(i,j)$ is said to be a \emph{non-candidate}.
\end{defn}

\begin{pro}[Candidate value]\label{pro:candidate_value}
  If $(i,j)$ is a candidate, then there is a unique value $d'_{i,j}$ to assign to the unknown entry $\Dm_{i,j}$ such that $\Dm(i,j)$ and $\Dm(i+1,j-1)$ have the same rank. This unique value is referred to as the \emph{candidate value}.
\end{pro}

\begin{proof}
  Since $\Dm(i+1,j-1)$ and $\Dm(i+1,j)$ have the same rank, the
  rightmost column of $\Dm(i+1,j)$ is a linear combination of the
  other ones.  For $\Dm(i,j)$ to have the same rank, its rightmost
  column should be expressed as the same linear combination of the
  other columns of $\Dm(i,j)$. Thus, $\Dm_{i,j}$ should be equal to the
  aforementioned linear combination of the entries on its left.
\end{proof}

Now, for each candidate $(i,j)$ on $\Dfar$, we compute the candidate
value (see Proposition~\ref{pro:candidate_value}) $d'_{i,j}$ for
$\Dm_{i,j}$ and compute $\gg_j^{-1}(d'_{i,j})$ which is a
\emph{predicted value} for $\efar$.

At this step, two situations may occur:
\begin{itemize}
\item either the prediction was true: $\Dm_{i,j} = d'_{i,j}$, in this case the candidate is said to be \emph{true};
\item or the prediction is wrong: $\Dm_{i,j} \neq d'_{i,j}$ and then
  $\rk \Dm(i,j) \neq \rk \Dm(i+1,j-1)$ which entails that $(i,j)$ is a
  discrepancy. In this case, the candidate is said to be \emph{false}.
\end{itemize}
Of course, when considering a given candidate, we cannot directly
guess whether it is true or false.  The key of the majority voting
technique rests on two facts:
\begin{enumerate}
\item True candidates give a true predicted value of $\efar$.
\item False candidates gives rise to new discrepancies while, from
  Proposition~\ref{prop:discrepancies}, the total number of
  discrepancies equals the rank of the matrix which is at most half
  the minimum distance.
\end{enumerate}
With the two above facts at hand, we deduce that there cannot be ``too many'' false candidates.
The statement to follow actually shows that a strict majority of candidates on the diagonal are true.
Thus, one can collect the predicted values for $\efar$ for any candidate and the one that
occurs in strict majority will be the actual value of $\efar$.

\begin{pro}\label{prop:majority_voting}
  Let $T$ be the number of true candidates on the diagonal $\Dfar$ and $F$ be the number of false ones,
  then \[T>F.\]
\end{pro}

\begin{proof}
Let $K$ denote the number of discrepancies
  that lie below $\Dfar$ (`$K$' stands for ``known discrepancies'').
  Using similar arguments as in the proof of
  Proposition~\ref{prop:discrepancies}, one shows that a position
  $(i,j)$ on $\Dfar$ such that $\Dm_{i,j}$ is conjugate to $\efar$
  is a candidate if and only if there is no known discrepancy on row $i$
  and on column $j$.
  Thus, since from Lemma~\ref{occurrences}, at least $d$ entries on $\Dfar$ are
  conjugates of $\efar$ we deduce that
\begin{equation}\label{candidatebound}
T+F = \#\text{candidates} \geq d-2K.
\end{equation}
Next, since false candidates yield discrepancies, from Proposition~\ref{prop:discrepancies}(\ref{item:rk}), we deduce that
\begin{equation}\label{eq:nr_disc}
  K+F \leq \rk (\Dm) = t.
\end{equation}
Recall that $t \le \lfloor \frac{d-1}{2} \rfloor$. Then, Equations
\eqref{candidatebound} and \eqref{eq:nr_disc} imply that
$T>F$.\end{proof}

\begin{rem}
  Note that since the number $F$ of false candidates is a nonnegative integer, Proposition~\ref{prop:majority_voting}
  asserts that $T>0$. In particular, the set of candidates is never empty.
\end{rem}

\noindent \textbf{In summary}, the decoding algorithm works as follows:
Given $Y = C+E$, with $C \in \RM_{\TH}(r,\nn)$
\begin{itemize}
\item Identify the known coefficients of $E$: they are the coefficients of $Y$ corresponding to monomials of $\TH$--degree $> r$;
\item Recover iteratively the unknown coefficients $e_\omega$ by decreasing index $\omega$ by applying the majority voting
  technique on the diagonal $\Delta_{\omega}$.
\end{itemize}
Algorithm~\ref{alg:dec} gives a precise description of this decoding procedure.

\begin{algorithm}
\caption{MajorityVote$(\cdot)$} \label{alg:mv}
\KwData{
  \begin{itemize}
  \item $\Am \in {\LL}^{N \times N}$ the $\GG$-Dickson matrix of rank $t$ of the error $E$, where some entries are unknown;
  \item $\omega \in [0, N-d]$
    such that all the entries of $\Am$ strictly below $\Delta_{\omega}$ are known;
  \item The list $L$ of positions on the diagonal $\Delta_{\omega}$ of
    $\mathbf{A}$ which contain a conjugate of $e_{\omega}$.
  \end{itemize}
}
\KwResult{The unknown coefficient $e_{\omega}$.}
\vspace{0.1cm}
\hrule
\vspace{0.1cm}

Candidates $\gets [\, ]$ \qquad /*An empty list */

\medskip

\For{$(i,j)$ in L}{
  \If{rank $\Am(i,j-1)$ = rank $\Am(i+1,j)$ = rank $\Am(i+1,j-1)$}{Add the pair $(i,j)$ to Candidates}
}

\medskip

Votes $\gets [\,]$\\

\medskip

\For{$C = (\omega +j, j)$ in \emph{Candidates}}{
  $\alpha\ \gets$ unique value for $\Am_{i,j}$ so that $\rk \Am(i+1,j-1) = \rk \Am (i,j)$ \hfill /* Proposition~\ref{pro:candidate_value} */ \\
$\text{Pred} \gets \gg_j^{-1}(\alpha)$ \qquad /* $C=(\omega+j, j)$, hence $e_{\omega} = \gg_j^{-1}(\Am_{i,j})$ */\\
Add $\text{Pred}$ to Votes
}

\medskip

\Return The element that occurs in Votes in strict majority
\end{algorithm}

\begin{algorithm}
\caption{RankRMDec$(\cdot)$}\label{alg:dec}
\KwData{$Y = \sum_{i}y_i \gg_i \in \LL[\GG]$, $r,\, \nn,\,N, \, t \le \lfloor\frac{d-1}{2} \rfloor$, where $d$ is the minimum distance of $RM_{\TH}(r,\nn)$.}
\KwResult{$E \in \LL[\GG]$ such that $\rk(Y -E)=t$.}
\vspace{0.1cm}
\hrule
\vspace{0.1cm}

$\Dm \gets$ $\GG$-Dickson matrix with unknown entries \hfill /* $\Dm$  will be the $\GG$ Dickson matrix of $E$ */\\
$L \gets $ indexes (in $[0, N-1]$) of unknown coefficients of $E$ \hfill /* They are the $i \in [0, N-d], $ */\\ \hfill /* such that $|\varphi^{-1}(i)| \leq r$. */\\
\hfill /* (See (\ref{phi}) and (\ref{eq:degree})  for the definitions). */\\
\For{$k \notin L$}{
$e_k \gets y_k$\\
\For{$0 \le i , j\le N-1$ such that $\gg_j \gg_k = \gg_i$}{
  $\Dm_{i,j} \gets \gg_j(e_k)$ \hfill /* Filling in $\Dm$ with known entries from $Y$ */
}
}

\For{$\omega \in L$ \textbf{in decreasing order}}{
$L_{\Delta_{\omega}} \gets$ positions $(\omega +j, j)$ in $\Delta_{\omega}$
such that $\Dm_{i,j}$ is conjugate to $e_{\omega}$\\
$e_{\omega} \gets \text{MajorityVote}(\Dm, \omega, L_{\Delta_\omega})$\\
\For{$0 \le i, j \le N-1$ such that $\gg_j \gg_{\omega} = \gg_i$}{
$\Dm_{i,j} \gets \gg_j(e_{\omega})$
}
}

\Return $E = \sum_i e_i \gg_i$
\end{algorithm}

\begin{rem}
  Actually, Algorithms~\ref{alg:mv} and~\ref{alg:dec} should be viewed
  as proofs of concept but should not be implemented as they are since
  Algorithm~\ref{alg:mv} involves too many independent rank
  calculations, \emph{i.e.} too many independent Gaussian eliminations
  on submatrices.  In order to get a good complexity, it is possible
  to perform a very similar algorithm that only requires to perform
  one Gaussian elimination on the whole matrix $\Dm$. This algorithm
  will be presented in Appendix~\ref{sec:appendix_voting} and will be
  the reference for the complexity analysis to follow.
\end{rem}

\subsection{Complexity}
As already mentioned, the way we described majority voting is not
fully efficient in terms of complexity since for any unknown
coefficient we should compute as many ranks as the number of
candidates. It is possible to avoid this cost by performing a
single Gaussian elimination process on a $N \times N$ matrix
in order to decode. The process is described in Appendix~\ref{sec:appendix_voting}.
This leads to the following statement.

\begin{thm}
  Let $\LL/\K$ be a degree $N$ Abelian extension with Galois group
  $\GG$ equipped with a system of generators
  $\TH=(\theta_1, \dots, \theta_m)$. Denote by
  $\nn = (n_1, \dots, n_m)$ the sequence of orders of the $\theta_i$'s
  in $\GG$.  Let $r \leq \sum_{i=1}^m (n_i-1)$ and $d$ denote the
  minimum distance of the code $\RM_{\TH}(r,\nn)$. Suppose
  we are given a primitive element $x$ of $\LL/\K$. Then,
  Algorithm~\ref{alg:dec} corrects any error pattern of weight
  $t \leq \frac{d-1}{2}$ in $\widetilde{\OO}(N^4)$ operations in $\K$.
\end{thm}

\begin{proof}
  As explained in Appendix~\ref{sec:appendix_voting} it is possible to
  perform the successive majority votings and to retrieve the error
  polynomial $E$ while performing a unique Gaussian elimination on an
  $N \times N$ matrix of rank $t$ together with $kN$ applications of Galois group
  elements. More precisely, Theorem~\ref{thm:non_naive_complexity}
  asserts that decoding costs $\OO (tN^2)$ operations in $\LL$ and
  $\OO (kN)$ applications of elements of $\GG$.

  From \cite[Cor.~11.11]{GG13}, operations in $\LL$ can be performed
  in $\widetilde{\OO}(N)$ operations in $\K$ as long as we are given a
  primitive element representation of $\LL/\K$. This gives a cost in
  $\widetilde{\OO}(tN^3) = \widetilde{\OO}(N^4)$ operations in $\K$
  for the Gaussian elimination.

  Next, any element of the Galois Group can be represented as an
  $N \times N$ matrix over $\K$ in the primitive element's basis
  $(1, x, x^2, \dots, x^{N-1})$. Such matrices can be pre-computed
  independently from the decoding. Then, the application of an element
  of $\GG$ can be performed in $\OO (N^2)$ operations in $K$. Thus, the
  calculation of $kN$ applications of Galois group elements has an
  overall cost in $\OO(kN^3)$ operations which, since $k \leq N$, is
  dominated by the $\widetilde{\OO}(N^4)$ cost of Gaussian elimination.
\end{proof}

\begin{rem}
  Actually the complexity can be made more precise and written as
  \[{\OO}(tN^3\log(N)\log\log(N) + k N^3).\]
\end{rem}

\subsection{Comparison with a previous work}
As already mentioned, a decoding algorithm was proposed in \cite{ACLN}
for rank metric Reed--Muller codes for $\nn = (n,n)$ but with a much smaller decoding radius.

\begin{exa}
  Consider $\GG= \ZZ/7\ZZ \times \ZZ/7\ZZ$, \emph{i.e.}, $\nn = (7,7)$
  and the code $\RM_{\TH}(4, \nn)$. This code has length $N = 49$,
  dimension $k = 15$ and minimum distance $d=21$.  According to
  \cite[Ex.~54]{ACLN}, the optimal choice consists in considering an
  error correcting pair with $a = 2$ and $b=5$ which permits to correct
  any error of rank $t\leq 6$ using the algorithm of \cite{ACLN}.

  Our algorithm corrects up to half the minimum distance \emph{i.e.}
  corrects any error of rank $\leq 10$.
\end{exa}

More generally, when $\GG = \ZZ / n\ZZ \times \ZZ/n\ZZ$ with
$n \rightarrow + \infty$ and $r \leq n$, \cite[Ex.~54]{ACLN} yields an
asymptotic analysis of the best decoding radius their algorithm can
achieve.  For $\rho = \lim_{n \rightarrow + \infty}\frac r n$ they can
correct about $(2 - \rho - \sqrt{3 - 2\rho})n^2$ errors while our
algorithm corrects up to half the minimum distance which asymptotically
corresponds to $(\frac{1-\rho}{2})n^2$ (under the assumption $r < n$).
The comparison with our algorithm is given in the Figure~\ref{fig:compare}.

\begin{figure}[h]
  \centering
      \begin{tikzpicture}[scale=0.75]      

        \pgfplotsset{every tick label/.append style={font=\normalsize}}
        
        \begin{axis}[
          xmin=0,
          xmax=1,
          ymin=0,
          ymax=0.7,
          xlabel={$\rho = \frac{r}{n}$},
          ylabel={radius $=\frac{t}{n^2}$},
          xtick={0.2, 0.4, ..., 0.8},
          xlabel style={anchor=north, at={(0.5,-0.08), font=\normalsize}},
          ylabel style={anchor=south, at={(0,0.5), font=\normalsize}},
          legend style= {anchor = north east, at={(0.985,0.985), font=\small}},
          cycle list name=mark list*
          ]
          
          \tikzstyle{my_style}=[domain=0:1, mark=none, samples=100]

          \addplot+[my_style, color=red]
          {(1-x)/2}; 
          \addlegendentry{This paper}

          \addplot+[my_style, color=black]
          {2-x-sqrt(3-2*x)}; 
          \addlegendentry{Decoding radius of \cite{ACLN}}
          
        \end{axis} 
    \end{tikzpicture}

    \caption{Comparison with \cite{ACLN} in the case
      $\GG = \ZZ /n \ZZ \times \ZZ/n\ZZ$ when $r<n$ and
      $n \rightarrow \infty$. The $x$-axis represents
      $\rho = \frac r n$ and the $y$--axis the relative decoding
      radius $\frac t {n^2}$.}
  \label{fig:compare}
\end{figure}

 \section{Conclusion and open questions }\label{sec:conclusion}

We give a deterministic decoding algorithm for rank metric Reed--Muller codes over an arbitrary Galois extension $\LL/\K$ which were introduced in \cite{ACLN} as $\LL$-linear subspaces of the skew group algebra $\LL[\GG]$ where $\GG = \Gal(\LL/\K)$. The decoding method can be seen as reconstruction of a $\TH$-polynomial, in particular, the error $\TH$-polynomial by recovering its unknown coefficients by majority vote for the unknown entries of the corresponding $\GG$-Dickson matrix. Our algorithm corrects any error pattern of rank up to half the minimum distance in $\widetilde{\OO} (N^{4})$ operations in $\K$, where $|\GG| = N$. We close with the following natural questions.

\begin{itemize}
\item  Is it possible to identify or construct rank metric codes (over finite or infinite fields) for which the majority voting method or minor cancellations of $\GG$-Dickson matrices allow to correct errors up to the unique decoding radius or beyond?
 
 \item Can the complexity of the decoding algorithm for rank metric Reed--Muller codes be improved from $\widetilde{\OO}(N^4)$?

\end{itemize}

\bibliographystyle{acm}

\appendix
\section{Minimizing the cost of Gaussian elimination while performing majority voting steps}\label{sec:appendix_voting}
If Algorithm~\ref{alg:dec} presented in
Section~\ref{sec:Dickson_decoding} corrects up to half the minimum
distance in polynomial time, it is actually not this efficient since
it includes many calls to Algorithm~\ref{alg:mv} which involves many
independent rank calculations and hence many independent uses of
Gaussian eliminations. In this appendix, we show how we can run mostly
the same algorithm while performing a single Gaussian elimination
process once for all.

\subsection{Context and setup}
We are given an Abelian extension $\LL/\K$ of degree $N$ and Galois
group $G$, the code $\RM_{\TH} (r, \nn)$ of dimension $k$ and a received $Y \in \LL[\GG]$
such that \[Y = C+E\] for some $C \in \RM_{\TH} (r, \nn)$ as in
Section~\ref{sec:Dickson_decoding} and $E \in \LL [\GG]$ with
$\rk (E) \leq t = \lfloor \frac{d-1}{2} \rfloor$ where $d$ denotes the
minimum distance of $\RM_{\TH} (r, \nn)$. Our objective is to compute
the matrix $\Dm = \Dm_{\GG}(E)$. It is already known from
Section~\ref{subsec:known_coefficients} that the entries of $\Dm$ are
partially known and, from Corollary~\ref{cor:known_coeffts}, all the
entries of $\Dm$ lying strictly below the diagonal $\Delta_{N-d}$ are
known. Unknown entries of $\Dm = \Dm_{\GG}(E)$ are written as formal
variables and will be iteratively specialised each time we discover a
new coefficient of $E$.

\noindent \textbf{Caution} Compared to Section
\ref{sec:Dickson_decoding}, where $\Dm$ always denotes the $\GG$-Dickson matrix $\Dm_{\GG}(E)$, in the present appendix, $\Dm$ will first denote this $\GG$-Dickson matrix with unknown entries written as formal variables. Then it will be modified by iteratively applying partial Gaussian elimination steps. Therefore, once we start running the algorithm, the matrix $\Dm$ will no longer be $\Dm_{\GG}(E)$ with
unknown variables. We chose to keep notation $\Dm$ by convenience.

\subsection{Elimination}
Given a row index $i_0$ and a column index $j_0$ corresponding to a known entry, we define the routine EliminateAbove which consists in
performing partial Gaussian elimination from the bottom by taking position $(i_0,j_0)$ as a pivot and eliminating any element above.

\begin{algorithm}
  \label{algo:elim_above}
  \caption{EliminateAbove($\Dm, i_0, j_0$)}
  \KwData{A matrix $\Dm$, a pivot position $(i_0,j_0)$ (\emph{i.e.}
    $\Dm_{i_0,j_0} \neq 0$).}
  \KwResult{Matrix $\Dm$ modified by partial elimination}

\medskip
  
  \For{$0 \leq i < i_0$}{
    $\text{Row}_i \gets \text{Row}_i - \Dm_{i,j_0} \Dm_{i_0,j_0}^{-1} \cdot \text{Row}_{i_0}$ \hfill /*  Rows of $\Dm$ are denoted
    ${(\text{Row}_i)}_{0 \leq i < N}$ */
  }
  \Return{$\Dm$}
\end{algorithm}

One easily observes that discrepancies and candidates are left
unchanged by EliminateAbove. What may change are candidate values.

\subsection{Finding the farthest unknown coefficient} Let us look at the
first step consisting in recovering the first furthest unknown coefficient
which, from Lemma~\ref{furthest}, is nothing but $e_{N-d}$.

The first step of the algorithm consists in starting a partial
Gaussian elimination from the bottom. Namely, for
$i = N-1, N-2, \ldots, N-d+1$ (taken by decreasing order), we consider
$\text{Row}_i$ of $\Dm$. If there are nonzero entries in this row and
lying below $\Delta_{N-d}$ then the leftmost nonzero entry is a pivot
with indexes $(i,j)$ for some $j < d$ and we run
EliminateAbove($\Dm, i,j$). Note that some row entries on the diagonal
and above are unknown and then some elimination operations are done
formally on variables.

Once this is done, then below the diagonal we get a few nonzero
rows whose leftmost nonzero entries are in distinct columns and these
positions are nothing but the known discrepancies. The situation is summarised
in Figure~\ref{fig:partial_elim}.

\begin{figure}[h]
  \caption{The situation after partial elimination below $\Delta_{N-d}$}
  \label{fig:partial_elim}
\begin{center}
  \begin{tikzpicture}
\draw[very thick] (0,0) rectangle (4,4);
\draw[thick, dashed] (0,2.8) -- (2.8,0);
\node[left, scale=.8] at (0,3) {$\Delta_{N-d}$};
\node[right, scale=.8] at (1.9,1.5) {$(i,j)~ \text{candidate}$};
\draw[gray, ->] (1.9,1.5) -- (1.48,1.4);
\fill (1.4,1.4) circle[radius=0.8 mm];
\fill (1.1,0.3) rectangle (1.25,0.45);
\draw[gray] (1.1,0.3) -- (4,0.3);
\draw[gray] (1.1,0.45) -- (4,0.45);

\fill (0.5,0.9) rectangle (.65,1.05);
\draw[gray] (0.5,0.9) -- (4,0.9);
\draw[gray] (0.5,1.05) -- (4,1.05);

\fill (0.2,1.8) rectangle (0.35,1.95);
\draw[gray] (0.2,1.8) -- (4,1.8);
\draw[gray] (0.2,1.95) -- (4,1.95);

\draw[gray, <-] (4,0.43) -- (5,1.4);
\draw[gray, <-] (4,1.9) -- (5,1.4);
\draw[gray, <-] (4,1.03) -- (5,1.4);
\node[right, scale=.8] at (5,1.4) {Rows with nonzero};
\node[right, scale=.8] at (5,1.1) {entries below $\Delta_{N-d}$};

\draw[gray, <-] (1.1,0.4) -- (-1.5, 1.4);
\draw[gray, <-] (0.5,1) -- (-1.5, 1.4);
\draw[gray, <-] (0.2,1.9) -- (-1.5, 1.4);
\node[left, scale=.8] at (-1.5,1.5) {Known};
\node[left, scale=.8] at (-1.5,1.23) {discrepancies};

\draw[dashed] (0,1.3) -- (1.5,1.3) -- (1.5,0); 
\draw[dashed] (0,1.5) -- (1.3,1.5) -- (1.3,0);
\node[scale=.5] at (.2,1.4) {$0 \ 0$};
\node[scale=.5] at (.75,1.4) {$\cdots \cdots$};
\node[scale=.5] at (1.15,1.4) {$0$};
\end{tikzpicture}
\end{center}
\end{figure}

Once this first partial elimination is done, one can observe that
\begin{enumerate}
\item Candidate positions on $\Delta_{N-d}$ are easy to identify: they correspond to positions $(i,j)$
  such that
  \begin{itemize}
  \item $\Dm_{\GG}(E)_{i,j}$ is conjugate to $e_{N-d}$;
  \item all the entries of $\Dm$ left to $(i,j)$ are zero;
  \item all the known discrepancies in $\Dm(i+1,j)$ are actually in
    $\Dm(i+1,j-1)$.
  \end{itemize}
\item The entry $\Dm_{i,j}$ equals $\gg_j(e_{N-d}) - \lambda_j$ where
  $\lambda_j$ is known since it comes from a linear combination of the
  entries lying strictly below $\Delta_{N-d}$. Moreover, since all the
  entries left to $\Dm_{i,j}$ are $0$, the natural candidate value is
  $0$ and hence the predicted value for $e_{N-d}$ is nothing but
  $\gg_j^{-1}(\lambda_j)$.
\end{enumerate}
Therefore, one can find $e_{N-d}$ by identifying candidate positions
on $\Delta_{N-d}$, then collecting predicted values $\gg_j^{-1}(\lambda)$ and, according to Proposition~\ref{prop:majority_voting},
the one that occurs in strict majority is $e_{N-d}$.

\subsection{Iteratively finding to other unknown coefficients}
Inductively, once a coefficient $e_\omega$ is found, we carry on Gaussian
elimination as follows:
\begin{itemize}
\item Substitute $e_{\omega}$ with the formal variable everywhere it
  occurs in $\Dm$;
\item All the candidates $(i,j)$ involved in the previous majority
  voting step that turned out to be false are new discrepancies and for
  them run EliminateAbove($\Dm, i,j$);
\item Let $\rho \in [0, \omega-1]$ be the index of the ``new'' farthest
  unknown coefficient, \emph{i.e.} the largest index of an unknown coefficient
  of $E$. Then for $i$ from $\omega - 1$ to $\rho + 1$ (by iteratively decreasing
  the value of $i$), if $\text{Row}_i$ contains nonzero elements below the diagonal
    $\Delta_\rho$ then let $j$ be the column index of the leftmost one and
    run EliminateAbove($\Dm, i,j$).
  \item Once partial elimination up to row $\rho 
  + 1$ is done a new majority
  voting process can be run in order to find $e_{\rho}$.
  This general majority voting process is described in Algorithm~\ref{algo:maj_vote_2}. 
\end{itemize}

\begin{algorithm}
  \caption{MajorityVoteWPE($\Dm, \omega$) \hfill  (``WPE'' stands for \emph{With Partial Elimination})}
  \label{algo:maj_vote_2}

  \KwData{Index $\omega$ for the diagonal; Matrix $\Dm$ partially eliminated below $\Delta_{\omega}$}
  \KwResult{Coefficient $e_{\omega}$}

\medskip
  
  Positions $\gets []$\\

  \medskip

\For{$j \in [0, N- \omega - 1]$}{
    \If{$\pinv(\omega+j) = \pinv(\omega)+\pinv(j)$}{
      Add $(\omega + j,j)$ to Positions. \hfill{/* Collect positions of conjugates of $e_{\omega}$ */}\\
      \hfill /* From Lemma~\ref{simple} they are the positions */\\
      \hfill /* $(\omega+j,j)$ such that $\Dm_{\GG}(E)_{\omega+j,j} = \gg_j(\omega)$  */
    }
  }
    \medskip

Values $\gets []$

    \medskip

    \For{$(\omega+j,j) \in \text{Positions}$}{ $\lambda_j \gets $ the
      element of $\LL$ such that $\Dm_{\omega+j, j}$ formally
      writes $\gg_j(e_{\omega}) - \lambda_j$.\\
      Add $\gg_j^{-1}(\lambda_j)$ to Values. \hfill /* Predicted value
      for $e_\omega$ */ }

\medskip
\Return{Element of Values occurring in strict majority}
\end{algorithm}

The whole decoding process is summarised in
Algorithm~\ref{algo:non_naive_decoding} below.

\begin{algorithm}
  \caption{Decoding algorithm with a single Gaussian elimination}
  \label{algo:non_naive_decoding}
  \KwData{Code $\RM_{\TH}(r,\nn)$ with minimum distance $d$; $Y \in \LL[\GG]$ such that $Y=C+E$
    for $C \in \RM_{\TH}(r,\nn)$ and $\rk E =t \leq \frac{d-1}{2}$.}
  \KwResult{The error $E$ as an element of $\LL[G]$}

\medskip
  
  \For{$i \in [0, N-1]$ such that $|\pinv(i)| > r$}{
    $e_i \gets y_i$ \hfill /* Collecting the known entries of $E$ */
  }

  \medskip

  $\Dm \gets $ Partial $\GG$-Dickson matrix for $E$ where unknown entries are formally
  written as ``$\gg_j(e_k)$''

  start $\gets N$\\
  far $\gets N-d$

  \medskip

\While{True}{
  
  \For{$i = \text{start}-1, \dots, \text{far}+1$ (in decreasing order)}{
    \If{$\Dm$ has nonzero entries below $\Delta_{\text{far}}$}{
      $j \gets$ column index of the leftmost nonzero entry in $\text{Row}_i$ below $\Delta_{\text{far}}$\\
      EliminateAbove($\Dm, i, j$)
    }
  }

  \medskip

  $e_{\text{far}} \gets$ MajorityVoteWPE($\Dm$, far)\\

  Substitute the value of $e_{\text{far}}$ everywhere it occurs formally in $\Dm$\\

  \medskip

  start $\gets$ far\\
  \If{Some entries of $E$ are still unknown}{
    far $\gets$ largest $i \in [0,N-1]$ such that $e_i$ is unknown
  }
  \Else{
    \Return {$E$}
  }
}
\end{algorithm}

\subsection{Complexity}\label{sec:app_complexity}
\subsubsection{Cost of elimination}
The procedure EliminateAbove costs $\OO (N^2)$ operations in $\LL$.
In the whole algorithm, the number of calls to EliminateAbove is
bounded by the number of discrepancies and hence by the rank of $E$.
While running Algorithm~\ref{algo:non_naive_decoding} some operations
are done only formally and postponed until the value of the unknown
coefficients is known. Still, it is as if we ran all the operations
required in $t$ successive calls of EliminateAbove but not in the
right order since some of them are postponed. Therefore, the overall
cost of Gaussian elimination in $\OO (tN^2)$ operations in $\LL$.

\begin{rem}
  Since elimination operations are done in a very specific order,
  we cannot invoke a possible use of fast linear algebra.
\end{rem}

\subsubsection{Cost of Galois action}
A call to MajorityVoteWPE involves $\OO(N)$ applications of an element
of $\GG$ (in order to compute the $\gg_j^{-1}(\lambda_j)$'s).  Once
the majority vote is done, the substitution of $e_{\omega}$ in $\Dm$
costs another $N$ applications of an element of $\GG$.  The number of
calls to MajorityVoteWPE is $k$ since $E$ has exactly $k$ unknown
coefficients when the algorithm starts, which leads to an overall
cost of $\OO (kN)$ evaluations of elements of $\GG$.

The following statement summarises the discussion on the complexity.

\begin{thm}\label{thm:non_naive_complexity}
  Algorithm~\ref{algo:non_naive_decoding} costs $\OO(tN^2)$ operations in
  $\LL$ and $\OO(kN)$ applications of elements of the Galois Group.
\end{thm}

\end{document}